\documentclass[12pt]{article}

\usepackage{times}
\usepackage{fullpage}
\usepackage{amsfonts,amssymb,amsmath,amsthm}
\usepackage{latexsym} 
\usepackage{gauss}
\usepackage{tikz}
\usepackage{tikz-cd}
\usepackage{calc}
\usepackage{url}
\usepackage{thmtools}
\usepackage{thm-restate}
\usepackage{hyperref}
\usepackage[capitalise, nameinlink]{cleveref}
\usepackage{algorithm}
\usepackage{algorithmicx}
\usepackage{algpseudocode}
\usepackage{graphicx}

\newcommand{\N}{\mathbb{N}}

\newcommand{\R}{\mathbb{R}}
\newcommand{\Z}{\mathbb{Z}}

\newcommand{\Acal}{\mathcal{A}}
\newcommand{\Ccal}{\mathcal{C}}
\newcommand{\Hcal}{\mathcal{H}}
\newcommand{\Lcal}{\mathcal{L}}
\newcommand{\Pcal}{\mathcal{P}}
\newcommand{\Rcal}{\mathcal{R}}
\newcommand{\Scal}{\mathcal{S}}
\newcommand{\Tcal}{\mathcal{T}}

\newcommand{\OR}{\mathrm{OR}}
\newcommand{\MAJ}{\mathrm{MAJ}}

\newcommand{\symvec}{\mathrm{symvec}}
\newcommand{\symmat}{\mathrm{Sym}}

\newcommand{\ind}{\mathrm{ind}}
\newcommand{\IP}{\mathrm{IP}}
\newcommand{\PARITY}{\mathrm{PARITY}}
\newcommand{\upto}{\mathbin{:}}
\DeclareMathOperator*{\argmin}{argmin}

\newcommand{\zeros}{\mathrm{zeros}}
\newcommand{\ones}{\mathrm{ones}}

\newcommand{\coisa}{weighted biadjacency matrix }

\newcommand{\floor}[1]{\lfloor #1 \rfloor}
\newcommand{\ceil}[1]{\left\lceil #1 \right\rceil}

\def\01{\{0,1\}}

\DeclareMathOperator*{\argmax}{arg\,max}

\newcommand{\ket}[1]{|#1\rangle}

\newtheorem{theorem}{Theorem}

\newtheorem{lemma}[theorem]{Lemma}
\newtheorem{corollary}[theorem]{Corollary}

\newtheorem{fact}[theorem]{Fact}

\theoremstyle{definition}

\newtheorem{definition}[theorem]{Definition}

\begin{document}
\title{Quantum algorithms for graph problems with cut queries}
\author{Troy Lee\thanks{Centre for Quantum Software and Information, University of Technology Sydney. Email: troyjlee@gmail.com} 
\and Miklos Santha\thanks{CNRS, IRIF, Universit\'e de Paris; Centre for Quantum Technologies and Majulab, National University of Singapore. Email: miklos.santha@gmail.com} 
\and Shengyu Zhang\thanks{Tencent Quantum Laboratory. Email: shengyzhang@tencent.com}}
\date{}
\maketitle

\begin{abstract}
Let $G$ be an $n$-vertex graph with $m$ edges.  When asked a subset $S$ of vertices, a cut query on $G$ returns the number 
of edges of $G$ that have exactly one endpoint in $S$.  We show that there is a bounded-error quantum algorithm that determines all connected components of $G$ after making $O(\log(n)^6)$ many cut queries. In contrast, it follows 
from results in communication complexity that any randomized algorithm even just to decide whether the graph is connected or not must make at least $\Omega(n/\log(n))$ many cut queries. 
We further show that with $O(\log(n)^8)$ 
many cut queries a quantum algorithm can with high probability output a spanning forest for $G$. 

En route to proving these results, we design quantum algorithms for learning a graph using cut queries.  We show 
that a quantum algorithm can learn a graph with maximum degree $d$ after $O(d \log(n)^2)$ many cut queries, and 
can learn a general graph with $O(\sqrt{m} \log(n)^{3/2})$ many cut queries. These two upper bounds are tight up to the poly-logarithmic factors,
and compare to $\Omega(dn)$ and $\Omega(m/\log(n))$ lower bounds on the number of cut queries needed by a randomized algorithm 
for the same problems, respectively.

The key ingredients in our results are the Bernstein-Vazirani algorithm, approximate counting with ``OR queries'', and learning 
sparse vectors from inner products as in compressed sensing.
\end{abstract} 

\newpage
\section{Introduction}
The \emph{cut} and \emph{additive} functions of an $n$-vertex undirected graph $G=(V,E)$ are $c: 2^{V} \rightarrow \N$ defined as $c(S) = |E(S, V \setminus S)|$, and $a: 2^{V} \rightarrow \N$ defined as 
$a(S) = |E(S, S)|$, where $E(S,T)$ is the set of edges between sets $S,T\subseteq V$.
In this paper we study quantum algorithms for several graph problems when the algorithm has oracle access to a cut or additive function for the graph.  In particular, we look at 
the problems of learning all the edges of a graph, determining if the graph is connected, outputting a spanning tree, and determining 
properties of the graph such as if it is bipartite or acyclic.

Motivation to study algorithms with cut or additive query access comes from at least two different sources.
Algorithms with a cut oracle have been studied for computing the minimum cut of a graph \cite{RSW18,MN20}, and
the study of graph algorithms with an additive oracle began in connection with 
an application to genomic sequencing \cite{GK98,ABKRS04}. We describe the previous works on these topics 
and their connection to our results in turn.

Rubinstein et al.\ \cite{RSW18} give a randomized algorithm that exactly computes the size of a minimum cut in an $n$-vertex unweighted and undirected graph with $\widetilde O(n)$ many cut queries  \footnote{The $\widetilde O()$ notation hides polylogarithmic factors in $n$.}. More recently, \cite{MN20} generalizes this result to also give a randomized $\widetilde O(n)$ cut query algorithm to exactly compute the size of a minimum cut in a weighted and undirected graph.  
These results are tight up to the polylogarithmic factors.  As observed by Harvey \cite{Harvey08}, known lower bounds on the communication complexity of determining if a graph is 
connected---when divided by $\log(n)$, the number of bits needed to communicate the answer of a cut query---give lower bounds on the number of cut queries needed by an algorithm to solve connectivity, 
and thus also min-cut.  It is known that the deterministic 
communication complexity of connectivity is $\Omega(n \log(n))$ \cite{HMT88} and the randomized communication complexity is $\Omega(n)$ \cite{BFS86}, which gives the tightness of the upper bounds for the cut query complexity claimed above.

Computing the minimum cut with a cut oracle is a special case of the problem of submodular function minimization 
with an evaluation oracle.  For an $n$-element set $\Omega$, a function $f : 2^\Omega \rightarrow \R$ is submodular if it satisfies $f(S \cap T) + f(S \cup T) \le f(S) + f(T)$ for all $S,T\subseteq \Omega$. 
As the truth table of $f$ is exponentially large in $n$, submodular functions are often studied assuming access to an \emph{evaluation oracle} for $f$, which 
for any $S \subseteq \Omega$ returns $f(S)$.  The submodular function minimization problem is to compute $\min_{S \subseteq \Omega} f(S)$.  The cut function is a submodular function,
 thus computing the minimum cut of a graph with a cut oracle is a special case of this problem \footnote{The cut function is a symmetric submodular function, $c(S) = c(V\setminus S)$.  For symmetric submodular functions 
 the minimization problem becomes to find a non-trivial minimizer $\emptyset \subset S \subset V$.}.  The ellipsoid method was originally used to 
show that submodular function minimization in general can solved in polynomial time with an evaluation oracle \cite{GLS81,GLS1988}, and the current record shows that this can be done with 
$O(n^3)$ many calls of the evaluation oracle \cite{Jiang20}, improving on the bound of $\widetilde O(n^3)$ by Lee, Sidford, and Wong \cite{LSW15}. If $M = \max_{A \subseteq V} |f(A)|$ then the current best known weakly polynomial algorithm makes  $O(n^2 \log(nM))$ many queries~\cite{LSW15}, and the best pseudo-polynomial time algorithm makes $O(nM^3 \log(n))$ many queries~\cite{CLSW17}.
The best lower bound on the number of evaluation oracle queries required to find the minimum value of a submodular function 
is $\Omega(n)$ for deterministic algorithms and $\Omega(n/\log(n))$ for randomized algorithms, coming from the above mentioned communication bounds on determining if a graph is 
connected.

Quantum algorithms for submodular function minimization have been studied for approximating $\min_{S \subseteq \Omega} f(S)$ \cite{HRRS19}.  For the problem of exact minimization no 
quantum algorithm has been considered.  Here we study the problem which is the source of the best known classical lower bounds on submodular 
function minimization, the problem of determining if a graph is connected with a cut oracle.  For this special case, we show that there is a surprisingly
efficient quantum algorithm: we give a quantum algorithm that determines if a graph is connected or not with high probability after $O(\log(n)^6)$ many cut queries (\cref{main:con}).  
With the same number of queries the algorithm can output the connected components of the graph, and can be extended to output a spanning forest after 
$O(\log(n)^8)$ many cut queries (\cref{thm:span_forest}).  
We leave deciding if minimum cut can also be computed by a quantum algorithm with a polylogarithmic number of cut queries as a tantalizing open problem.

Additive queries were originally defined because of an application to genomic sequencing.  One technique for sequencing a genome is to 
first use shotgun sequencing to come up with random segments of DNA called \emph{contigs}.  To complete the sequencing of the genome it remains to 
fill the gaps between the contigs, and thus also figure out how the contigs are connected to one another.
This can be done through polymerase chain reaction (PCR) techniques.  To do this, one makes primers, short sequences 
that pair to the end of a contig.  When two contigs are connected by a gap and the primers for each end of this gap are placed in solution with them, the PCR technique 
can connect the contigs by filling in the gap of base pairs between them.  If there are $n$ contigs, the process of putting separately all pairs of primers 
together with them to find the contigs connected together can take $\binom{n}{2}$ PCR experiments.  This complexity led to the development of multiplex PCR, where one 
places many primers together with the contigs at once.  The algorithmic question becomes, what is the best way to add primers to the contigs in order to 
fill all gaps between contigs while minimizing the number of multiplex PCR assays? 

Grebinski and Kucherov nicely formalize this algorithmic problem in the language of graph theory \cite{GK98}.  
They give several possible models, depending on exactly what information one assumes to learn from a multiplex PCR assay.  
The additive query arises if one assumes that when primers are placed together 
with contigs, one can read from the result the number of contigs that were paired.  A weaker model only assumes that when primers are 
placed together with contigs, one learns whether or not there is a pairing of contigs.  In this model, when one queries a 
subset of vertices $S$, one learns whether or not $|E(S,S)| > 0$.  We call these \emph{empty subgraph queries}.

In the genomic sequencing application, the primary goal is to learn all the edges in the graph.  Furthermore, for this application the graph is 
guaranteed to have degree at most $2$. Grebinski and Kucherov \cite{GK98} show that a Hamiltonian cycle can be learned by a randomized algorithm 
with $O(n\log n)$ many empty subgraph queries, and \cite{ABKRS04} gives a nonadaptive randomized algorithm to learn a matching with $O(n \log n)$ many empty subgraph queries.
After these works, the complexity of learning general weighted graphs with additive queries was also extensively studied
\cite{GK00, CK10,Maz10,BM10,BM11}.  Notable results are that a graph of maximum degree $d$ can be learned with $O(dn)$ many additive queries \cite{GK00}, and 
that an $m$-edge graph can be learned with $O(\frac{m \log n}{\log m})$ many additive queries \cite{CK10,BM11}.  Both of these results are tight up to a logarithmic factor in $n$.  

We show the additive quantum query complexity of learning a graph with maximum degree $d$ is 
$\Theta(d \log(n/d))$ 
(see \cref{cor:add_learn} and \cref{cor:lower_d}), 
and that a graph with $m$ edges can be learned after $O(\sqrt{m \log n} + \log(n))$ many cut queries (see \cref{cor:add_learn}), which is tight up to logarithmic factors (see \cref{cor:lower_m}).  
It is straightforward to show that a cut query can be simulated by 3 additive queries.  Somewhat surprisingly, we show that in general for weighted graphs simulating an 
additive query can require $\Omega(n)$ many cut queries (see \cref{lem:cut_lower}).  Nonetheless, our algorithms for learning graphs can also be made to work with cut queries, at the expense of 
an additional multiplicative logarithmic factor.  The quantum algorithm for efficiently learning a low degree graph with cut queries is a key subroutine in our algorithm for connectivity.

\subsection{Our techniques}
We begin by giving an overview of our quantum algorithm for learning a graph of maximum degree $d$.  The high level idea follows the $O(dn)$ classical algorithm given in \cite{GK00}.  
The starting point is the basic principle of compressed sensing: a sparse vector can be learned from its dot product with a few random vectors.  Say $y \in \{0,1\}^n$ has at most $d$ ones.  
With high probability, we can learn $y$ from the values of $x_1^T y, \dots, x_k^Ty$ for $k = 3d \log n$ where each $x_i \in \{0,1\}^n$ is chosen randomly (see \cref{lem:ran_red}).  
Moreover, this algorithm is non-adaptive, thus given a matrix $Y \in \{0,1\}^{n \times n}$ where each column has at most $d$ ones, we can learn $Y$ with high probability from the product $XY$ where 
$X \in \{0,1\}^{k \times n}$ is chosen randomly with $k = 3d \log n$.  We can apply this principle to learn the adjacency matrix $A_G$ of a graph $G$ of maximum degree $d$.  

Classically, a single product $x^T A_G$ can be computed with $O(n)$ cut queries.  We will show how to use the Bernstein-Vazirani algorithm \cite{BV97}, 
to compute $x^T A_G$ with a constant number of cut queries, provided the graph $G$ is bipartite.  We can then apply this to learn a general graph $G$ as the complete graph
can be covered by $O(\log n)$ many complete bipartite graphs.  This leads to an overall complexity of $\widetilde O(d)$ many cut queries.

These results can be phrased more generally in terms of learning an $m$-by-$n$ matrix $A$ when given the ability to query $x^TAy$ for Boolean vectors $x \in\{0,1\}^m,y \in \{0,1\}^n$.  
We call such queries \emph{matrix cut queries}, and they provide a very clean way of formulating our results \footnote{We term these matrix cut queries because of the relation to the cut 
norm of a matrix $A$, defined as $\max_{x \in \{0,1\}^m, y \in \{0,1\}^n} |x^T A y|$.  The cut norm played a crucial role in the matrix decomposition results of~\cite{FK99a} used for efficient approximation 
algorithms for maximum cut and other graph problems.}
Such queries are also used in the classical learning results \cite{GK00,Maz10, BM11}, and 
we find it useful to make them explicit here.

Our algorithm for connectivity is a contraction based algorithm.  If we find an edge between vertices $u$ and $v$, we can merge $u$ and $v$ into a single vertex without changing 
the connectivity of the graph.  More generally, if we know that a set $S_1$ of vertices is connected, and a disjoint set $S_2$ of vertices is connected, and we learn that there is an edge between a vertex 
in $S_1$ and a vertex in $S_2$, we can merge $S_1$ and $S_2$ without changing the connectivity of the graph.  

Our algorithm for connectivity proceeds in rounds, and maintains the invariant of having a partition of the vertex set $V$ into sets $S_1, \ldots, S_k$, each of which is known to be connected.  
We call such sets $S_i$ supervertices, and the number of $S_j$ such that there is an edge between a vertex in $S_i$ and a vertex in $S_j$ the \emph{superdegree} of $S_i$. 
A round proceeds by first in parallel approximating the superdegree of each $S_i$.  For this we use ideas originating with Stockmeyer \cite{Stockmeyer83} for approximately counting 
the number of ones in a vector $y \in \{0,1\}^n$ with OR queries, queries that for a subset $S \subseteq \{1,\ldots, n\}$ answer $1$ iff there is an $i \in S$ with $y_i =1$.  
We show that we can efficiently simulate the OR queries needed in the context of approximating the superdegree by cut queries.
For supervertices whose superdegree is below a threshold---taken to be $\Theta(\log(n)^2)$---we learn all of their superneighbors with polylogarithmically many cut queries using our techniques for learning a matrix with sparse columns.  
To take care of the supervertices with superdegree above the threshold, we randomly sample $k/2$ many of the $S_j$.  With high probability every high superdegree supervertex will be connected 
to at least one $S_j$ in the sample set, and moreover we show that we can learn edges witnessing this fact with polylogarithmically many cut queries.  We then contract the supervertices according to all 
the edges learned in the round, and show that we either learn that the graph is disconnected or arrive at a new partition of $V$ into at most $k/2$ many connected sets.  In this way, the algorithm terminates 
in at most $O(\log n)$ many rounds, and as each round uses polylogarithmically many cut queries we arrive at the polylogarithmic complexity.

\subsection{Related work}
As far as we are aware, ours is the first work studying quantum algorithms using a cut or additive
query oracle.  Quantum algorithms for graph problems typically use the adjacency matrix or adjacency list 
input model.  Early work \cite{DHHM06,SYZ04,Zha05}\ gives tight bounds for many graph problems in these models.
For example, it is shown in \cite{DHHM06} that the quantum query complexity of connectivity is $\Theta(n^{3/2})$ in the adjacency 
matrix model and $\Theta(n)$ in the adjacency list model, the quantum query complexity of minimum spanning tree is $\Theta(n^{3/2})$ in the adjacency matrix model and $\Theta(\sqrt{nm})$ in the adjacency list model, for a graph with $m$ edges. Focusing on the adjacency matrix model, the quantum query complexity of Bipartiteness and of Graph Matching is $\Omega(n^{3/2})$ \cite{Zha05}, and that of Scorpion graph is $\tilde \Theta(\sqrt{n})$ , which is the lowest possible quantum query complexity for total graph properties \cite{SYZ04}. Also see \cite{MSS07,Belovs12,JKM13,LMS17,LeG14} for triangle finding, \cite{CK12} for general graph minor-closed graph properties, and \cite{BDCG+20} for discussions on partial graph properties.

In the time complexity model, recent work 
of Apers and de Wolf \cite{AdW19} shows that a cut of size at most $(1+\epsilon)$-times
that of the minimum 
cut can be found in time $\widetilde O(\sqrt{mn}/\epsilon)$ given adjacency list access to the graph.  
Their work more generally shows that an $\epsilon$-cut sparsifier of a graph with $\widetilde O(n/\epsilon^2)$ many 
edges can be found in time $\widetilde O(\sqrt{mn}/\epsilon)$. 
Originally defined in~\cite{BK96}, an $\epsilon$-cut sparsifier of a graph $G$ is a reweighted subgraph $H$ such that the 
value of every cut in $H$ is a multiplicative $(1+ \epsilon)$-approximation of the corresponding cut value in $G$.  

In the classical setting, a recent work~\cite{RWZ20} studies a generalization of matrix cut queries 
where for a fixed field $\mathbb{F}$ one can access an input matrix $A \in \mathbb{F}^{m \times n}$ through queries of the form $x^T A y$ 
for $x \in \mathbb{F}^m, y \in \mathbb{F}^n$.  They examine the complexity of an assortment of problems from graph theory, linear 
algebra, and statistics in this model.

\subsection{Organization}
The rest of the paper is organized as follows. In \cref{sec:pre}, we set up the model and notation, and give some algorithmic ingredients such as degree estimation by OR queries, and a quantum primitive of vector learning that is used a number of times in later sections. In \cref{sec:matrix_learning} and \cref{sec:graph_learning} we introduce several oracles for accessing a matrix or a graph, and discuss their relative powers. Some intermediate results built on these oracles will be presented, which lead to the cut oracle algorithms in the next sections:  \cref{sec:connect} for computing connected components by $O(\log(n)^6)$ queries, and \cref{sec:span} for computing spanning forest by $O(\log(n)^8)$ cut queries. Two applications of the latter algorithm will also be given, which test if a graph is bipartite or acyclic, also with $O(\log(n)^8)$ many cut queries. 
The paper concludes with several open problems. 

\section{Preliminaries and primitives} \label{sec:pre}
For a positive integer $M$, we denote $\{0,1,2,\ldots, M-1\}$ by $[M]$.
For two vectors $X, Y \in [M]^k$, we denote their dot product over the integers by
$X \cdot Y = \sum_{i=1}^k X_iY_i.$  For a set $S \subseteq \{1, \ldots, M\}$ we denote 
the complement of $S$ by $\bar S$.  For a set $U$, we let $\tilde U = \{ \{u\} : u \in U\}$.  For a string $x \in \{0,1\}^n$ we use $|x|$ for the Hamming 
weight of $x$, i.e.\ the number of ones.  Let $\OR_n : \{0,1\}^n \rightarrow \{0,1\}$ denote the OR function, 
i.e.\ the function such that $\OR_n(x) = 1$ iff $|x| > 0$.  Let $\MAJ_n : \{0,1\}^n \rightarrow \{0,1\}$ denote the Majority 
function, i.e.\ the function such that $\MAJ_n(x) =1$ iff $|x| \ge \ceil{n/2}$.  When the input length is clear from context 
we will drop the subscript.

In pseudocode for our algorithms we will use some Matlab-like notation.
We use $\zeros(k,\ell)$ and $\ones(k,\ell)$ to denote the $k$-by-$\ell$ all zeros matrix and all 
ones matrix, respectively.   For $M$ a $k$-by-$\ell$ matrix and 
$H \subseteq \{1, \ldots, k\}, R \subseteq \{1, \ldots, \ell\}$ we use 
$M(H,R)$ to refer to the $|H|$-by-$|R|$ submatrix of $M$ given by selecting the 
rows of $M$ in $H$ and columns of $M$ in $R$.  We will further use the shorthand 
$M(i\colon j, k\colon \ell)$ for $M(H,R)$ where $H = \{i, i+1, \ldots,j\}$ and $R = \{k, k+1, \ldots, \ell\}$.  
To denote the $i^{th}$ column of $M$ we will use $M(:,i)$.  
Similarly for a vector $x \in \mathbb{R}^k$, 
we use $x(H) \in \mathbb{R}^{|H|}$ to denote the vector formed by selecting the coordinates in $H$.  
For column vectors $x \in \R^k, y \in \R^\ell$ we use $[x;y] \in \R^{k+\ell}$ for the column vector 
formed by vertically concatenating them. For a vector $x \in \mathbb{R}^k$ and $d \leq k$, we say that
$x$ is $d$-sparse if at most $d$ coordinates in $x$ are nonzero. We use $\ell_0(x)$ to denote the sparsity, i.e.\ the number of nonzero entries, of vector $x$.

We will frequently encounter the situation where we have a set $\Scal = \{S_1, \ldots, S_k\}$ 
and take a subset $H \subseteq \Scal$ where $H = \{S_{i_1}, \ldots, S_{i_t}\}$.  
We use $\ind(H) = \{i_1, \ldots, i_t\}$ to return the indices of the elements in $H$.

\begin{lemma}
\label{lem:ran_red}
Let $r,M,d$ be positive integers with $M \ge 2$ and $d \le r/2$.  
Let $N \le d \binom{r}{d}M^d$ be the number of
$d$-sparse strings in $[M]^r$.  Let 
$A$ be an $N$-by-$r$ matrix whose rows are all the $d$-sparse strings in $[M]^r$.  Let $R$
be a random $r$-by-$q$ Boolean matrix, with each entry chosen independently and uniformly from 
$\{0,1\}$.  If $q = \ceil{2d\log(eMr/d) + 2\log(d) + \log(1/\delta)}$ then 
$AR \bmod M$ will have distinct rows with probability at least $1-\delta$.
\end{lemma}

\begin{proof}
Let $A_i, A_j$ be two different rows of $A$.  Then for a random vector $z \in \{0,1\}^r$,  
$\Pr_z[A_i z \bmod M = A_j z \bmod M] \leq 1/2$.  Therefore 
$\Pr_R[ A_i R \bmod M = A_j R \bmod M] \leq 1/2^q$.  The result follows by a union bound over 
the $\binom{N}{2}$ pairs of distinct rows.
\end{proof}

\begin{definition}[OR query]
For $x \in \{0,1\}^\ell$ and $S \subseteq \{1, \ldots, \ell\}$ an OR query $\OR(x,S)$ returns $\vee_{i \in S} x_i$.  
For strings $x^{(1)}, \ldots, x^{(k)} \in \{0,1\}^\ell$ and a subset $S \subseteq \{1, \ldots, \ell\}$, a $k$-OR query 
returns the string $a \in \{0,1\}^k$ where $a_i = \OR(x^{i},S)$.  
\end{definition}

The problem of estimating the Hamming weight of a string $x$ using OR queries was considered in the seminal work of Stockmeyer \cite{Stockmeyer83}.  We recount 
his basic analysis here, modifying it for our application of using $k$-OR queries to estimate the Hamming weight of $x^{(1)}, \ldots, x^{(k)}$ in parallel.

\begin{fact}
\label{fact:exp}
For all $x \ge 1$
\[
\left(1-\frac{1}{x}\right)^x < \frac{1}{e} < \left(1- \frac{1}{x}\right)^{x-1}
\]
\end{fact}

\begin{definition}
Let $x \in \{0,1\}^\ell$.  An $r$-out-of-$\ell$ sample consists of uniformly at random choosing $r$ many elements of $\{1,\ldots, \ell\}$ with replacement.  An $r$-test
consists of taking an $r$-out-of-$\ell$ sample $S$ and querying $\OR(x,S)$.  
When $\OR(x,S) = 1$, we say that the $r$-test $S$ succeeds.
\end{definition}

\begin{lemma}
\label{lem:ell_test}
Let $x \in \{0,1\}^\ell$ and suppose that $|x| = t$.  The probability that an $r$-test $S$ succeeds is
\[
\Pr_S[\OR(x,S) = 1] = 1 - \left(1- \frac{t}{\ell}\right)^r \enspace,
\]
where the probability is taken over the choice of an $r$-out-of-$\ell$ sample $S$.  
In particular,
\[
1-\exp\left(-\frac{r t}{\ell}\right) < \Pr_S[\OR(x,S) = 1] < 1- \exp\left(-\frac{r t}{\ell-t}\right)
\]
\end{lemma}

\begin{proof}
As $S$ is chosen with replacement, the probability that any element $i \in S$ satisfies $x_i = 1$ is exactly $\frac{t}{\ell}$.  The first statement 
follows accordingly.  The second statement follows by applying \cref{fact:exp} to the first statement.
\end{proof}

\begin{algorithm}
\caption{Approximate Count$(\delta)$}
\label{alg:parallel_count}
 \hspace*{\algorithmicindent} \textbf{Input:} $k$-OR oracle for strings $x^{(1)}, \ldots, x^{(k)} \in \{0,1\}^\ell$, error bound $\delta$ \\
 \hspace*{\algorithmicindent} \textbf{Output:} A vector $b \in \mathbb{R}^k$ that is a good estimate of $(|x^{(1)}|, \ldots, |x^{(k)}|)$ with probability $1-\delta$.
\begin{algorithmic}[1]
\State $a \gets 200 \ceil{\log\left( \frac{k (\ceil{\log \ell}+1)}{\delta} \right)}$
\For{$j=0$ to $\ceil{\log \ell}$}
  \For{$q = 1$ to $a$}
    \State $S_q \gets$ randomly choose $\min(2^j, \ell)$ many elements from $\{1, \ldots, \ell\}$ with replacement
    \State $\mathrm{ans}_q \gets (\OR(x^{(1)}, S_q), \ldots, \OR(x^{(k)}, S_q))$
  \EndFor
  \For{$i=1$ to $k$}
    \State $B(i,j) \gets \MAJ(\mathrm{ans}_1(i), \ldots, \mathrm{ans}_a(i))$
  \EndFor
\EndFor
\For{$i=1$ to $k$}
  \State $s \gets \argmin_j \{B(i,j) = 1\}$
  \State $b(i) \gets \frac{\ell}{2^s}$
\EndFor
\end{algorithmic}
\end{algorithm}

\begin{definition}[Good estimate]
\label{def:good}
We say that $b \in \R^k$ is a \emph{good estimate} of $c \in \R^k$ if $b(i)/4 \le c(i) \le 2 b(i)$ for all $i \in \{1, \ldots, k\}$.
\end{definition}

Next we see how to use $k$-OR queries to estimate the Hamming weights of $k$ given strings.
\begin{lemma}
\label{lem:app_count_analysis}
Let $x^{(1)}, \ldots, x^{(k)} \in \{0,1\}^\ell$.
Taking $a = 200 \ceil{\log\left( \frac{k (\ceil{\log \ell}+1)}{\delta} \right)}$, \cref{alg:parallel_count} outputs a vector $b$ that is a good estimate of $(|x^{(1)}|, \ldots, |x^{(k)}|)$ 
with probability at least $1-\delta$ after making $O\left((\log(\ell)+1)\log\left( \frac{k (\log(\ell)+1)}{\delta} \right)\right)$ many $k$-OR queries.  
\end{lemma}

\begin{proof}
First we argue about the complexity.  There are $\ceil{\log \ell}+1$ iterations of the outer for loop, and inside the loop we make $a$ many $k$-OR queries.  
Thus the total number of $k$-OR queries is 
\[
O(a (\log(\ell)+1)) = O\left((\log(\ell)+1)\log\left( \frac{k (\log(\ell)+1)}{\delta} \right)\right) \enspace.
\]

Now we argue about the error probability.  Fix some $x^{(i)}$, which will henceforth be called $x$.  Say the estimate the algorithm gives 
for $|x|$ is $\alpha$.  We will upper bound the probability that $\alpha$ is a bad estimate for $|x|$, and then use a union bound over the $k$ many strings to 
get the final result.

Suppose that $r=2^j$ is such that $r \ge \ell/|x|$.  Then if $S$ is a $r$-out-of-$\ell$ sample, by \cref{lem:ell_test} 
\[
\Pr_S[\OR(x,S) = 1] > 1 - \exp\left( -\frac{r |x|}{\ell} \right) \ge 1- \frac{1}{e} > 0.6 \enspace.
\]
Thus by a Chernoff bound, the probability that the majority of $a$ many $r$-tests is $0$ is at most $\exp(-a/200)$.  As the 
algorithm must perform an $r$-test for an $r$ satisfying $\ell/|x| \le r \le 2\ell/|x|$, this means our estimate $\alpha$ will satisfy $|x| \le 2\alpha$, 
except with probability at most $\exp(-a/200)$.

To bound the probability that our estimate is too large, we first need to treat the special case where
$|x| \ge \ell/2$.  In this case, the probability that the majority of $a$ many $2$-tests is $1$
is at least $1-\exp(-a/200)$.  The algorithm is therefore correct in this case with at least this probability as it outputs 
a valid answer if either a majority of $1$-tests or a $2$-tests succeeds.

For the remainder of the proof, therefore, we assume that $|x| < \ell/2$.  Suppose that $r=2^j$ is such that $\ell/r > 4 |x|$.  
Then if $S$ is an $r$-out-of-$\ell$ sample, by \cref{lem:ell_test} 
\[
\Pr_S[\OR(x,S) = 1] < 1 - \exp\left( -\frac{r |x|}{\ell-|x|} \right) \le 1- \exp\left(-\frac{2r |x|}{\ell} \right) \le  1- \exp(-1/2) \le 0.4 \enspace.
\]
Thus by a Chernoff bound, the probability that the majority of $a$ many $r$-tests is $1$ is at most $\exp(-a/200)$.
By a union bound, the probability that an $r$-test will be $1$ for any $r=2^j$ with $r < \ell/(4|x|)$ is at most 
$\ceil{\log \ell} \exp(-a/200)$.  Thus, our estimate $\alpha$ will satisfy $\alpha/4 \le |x|$ except with probability $\ceil{\log \ell} \exp(-a/200)$.
Thus by a union bound, overall we will produce an estimate $\alpha$ satisfying $\alpha/4 \le |x| \le 2\alpha$ except with probability 
$(\ceil{\log \ell} +1)\exp(-a/200)$.

Finally, by a union bound over the $k$ many strings the total failure probability will be at most $k (\ceil{\log \ell}+1) \exp(-a/200)$, giving the 
lemma.
\end{proof}

\begin{lemma}
\label{lem:sample}
Let $x^{(1)}, \ldots, x^{(k)} \in \{0,1\}^\ell$ be such that $t/8 \le |x^{(i)}| \le 2t$ for all $i \in \{1, \ldots, k\}$.
For $\delta > 0$, sample with replacement $\frac{8 \ell \ln(k/\delta)}{t}$ elements of $\{1, \ldots, \ell\}$, and call the resulting set $R$. Then 
\begin{itemize}
\item $\Pr_R[ \exists i \in \{1, \ldots, k\}: |x^{(i)}(R)| = 0] \le \delta$,
\item $\Pr_R[ \exists i \in \{1, \ldots, k\}: |x^{(i)}(R)| > 64 \ln(k/\delta)] \le \delta$.
\end{itemize}
\end{lemma}

\begin{proof}
Let $x \in \{0,1\}^\ell$ with $t/8 \le |x| \le 2t$.
The probability that $|x(R)| = 0$ is at most $\delta/k$ by \cref{lem:ell_test}.  Thus the first item holds by a union 
bound over the $k$ many strings $x^{(1)}, \ldots, x^{(k)}$.

For the second item, note $\ln(k/\delta) \le \mathbb{E}_R[|x(R)|] \le 16 \ln(k/\delta)$.  By a Chernoff bound, the probability 
that $|x(R)|$ is a factor $c$ larger than its expectation is at most $\exp(-(c-1) \ln(k/\delta)/3)$.  Thus taking $c =4$ the probability 
that $|x(R)| > 64 \ln(k/\delta)$ is at most $\delta/k$.  The second item then holds by a union 
bound over the $k$ many strings $x^{(1)}, \ldots, x^{(k)}$.
\end{proof}

Next we give our quantum primitive algorithm for learning a vector from subset sums. Note that different than state learning in many previous work, here we aim to learn all entries of the vector precisely and correctly. This algorithm will be repeatedly used in various forms in later sections. 

\begin{lemma}
\label{lem:QFT}
Let $x \in [M]^k$ and suppose we have an oracle that for any subset $S \subseteq [k]$ returns 
$\sum_{i \in S} x_i \bmod M$.  Then there is a quantum algorithm which learns $x$ with $m = \ceil{\log(M)}$ queries 
without any error.  
\end{lemma}

\begin{proof}
We can represent all the elements of $[M]$ using $m$ bits. 
We first describe how we can use the oracle to compute $t \cdot x \bmod M$ for any $t \in [M]^k$.  
Let $t = (t_1, \ldots, t_k)$ and let $t_i = \sum_{j=0}^{m-1} 2^j t_i(j)$, where $t_i(j) \in \{0,1\}$.  Then 
\[
t \cdot x \bmod M = \sum_{j = 0}^{m-1} 2^j ( t(j) \cdot x \bmod M) \bmod M
\]
where $t(j) = (t_1(j), t_2(j), \ldots, t_k(j)) \in \{0,1\}^k$.  Now $t(j) \cdot x \bmod M$ can be computed with one call to the oracle, thus 
$t \cdot x \bmod M$ can be computed with $m$ calls to the oracle.  

With a quantum algorithm we can do this in superposition over all $t$. Combining this with the phase kickback trick
allows us to compute the inverse quantum Fourier transform of $x$ over 
$(\Z_M)^k$.  
Similarly to the Bernstein-Vazirani algorithm, we can then learn $x$ by applying the quantum Fourier transform. 

More precisely, we will work with two quantum registers. In the first register we are computing over the group
$(\Z_M)^k$, and in the second register over $\Z_M$. We will use in the second register
the auxiliary state
$$
\ket{\xi_M} = \frac{1}{\sqrt M} \sum_{j=0}^{M-1} \omega_M^j \ket{j},
$$
where $\omega_M$ is a primitive $M$th root of the unity. In the first register we create the uniform superposition,
therefore we start with the state
$$
\frac{1}{\sqrt {M^k}} \sum_{t \in [M]^k} \ket{t} \ket{\xi_M}.
$$
With $m$ queries to the oracle we compute $t \cdot x \bmod M$ that we add to the second register, creating
$$
\frac{1}{\sqrt {M^k}} \sum_{t \in [M]^k} \omega_M^{- t \cdot x} \ket{t} \ket{\xi_M}.
$$
Applying the quantum Fourier transform over $(\Z_M)^k$, we can find in the first register $x.$
Overall the quantum algorithm uses $m$ queries.
\end{proof}

\section{Learning a matrix}
\label{sec:matrix_learning}
Let $n$ be a positive integer, and $V$ an ordered set of size $n$.
For a subset $S$ of $V$, we denote by 
$\chi_S \in \{0,1\}^n$ is the characteristic vector of $S$.

\begin{definition}[Matrix cut oracle]
Given a matrix $A \in \N^{k \times \ell},$ the matrix cut oracle for $A$ is the function 
$m_A : \{0,1\}^k \times \{0,1\}^\ell \rightarrow \N$ satisfying $m_A(x,y) = x^T A y$.  
\end{definition}

\begin{lemma}
\label{lem:Ay}
Let $\alpha, M \in \N$ and $A \in [\alpha]^{k \times \ell}$ be a matrix.  There is a quantum algorithm making $\ceil{\log(M)}$ many matrix cut queries 
to $A$ that perfectly computes $(Ay) \bmod M$ for any $y \in \{0,1\}^\ell$.
\end{lemma}

\begin{proof}
Let $S \subseteq \{1, \ldots, k\}$.  For any $y \in \{0,1\}^\ell$ we have $\sum_{i \in S} (Ay)_i \bmod M = \chi_S^T A y \bmod M$,
thus we can compute $\sum_{i \in S} (Ay)_i \bmod M$ with one matrix 
cut query.  The result then follows from \cref{lem:QFT}.
\end{proof}

\begin{corollary}
\label{cor:individual}
Let $A \in [M]^{k \times \ell}$.  There is a quantum algorithm that perfectly learns $A$ after $\min\{k,\ell\} \cdot \ceil{\log(M)}$ 
many matrix cut queries.
\end{corollary}

\begin{proof}
If $\ell \le k$ define $B = A$, otherwise let $B = A^T$.  Let $m = \min\{k,\ell\}$ be the number of columns in $B$.  
For each standard basis vector $e_1, \ldots, e_m$, we can learn $Be_i = Be_i \bmod M$ with $\ceil{\log(M)}$ matrix cut queries by 
\cref{lem:Ay}.  This tells us the $i^{th}$ column of $B$.  Thus we can learn $B$ entirely with $m \ceil{\log(M)}$ many 
matrix cut queries.
\end{proof}

\begin{lemma}
\label{lem:learn_matrix}
Let $A \in [M]^{k \times \ell}$ be a matrix with $d$-sparse rows.
There is a quantum algorithm that 
learns $A$ with probability at least $1-\delta$ after making 
$(4d\ceil{\log(M \ell/d)} + \ceil{\log(1/\delta)}) \ceil{\log(M)}$ 
many matrix cut queries.
\end{lemma}

\begin{proof}
If $d \ge \ell/2$ then we use \cref{cor:individual} to learn $A$ perfectly with $\ell \ceil{\log(M)} \le 4 d\ceil{\log(M \ell/d)}$ 
many matrix cut queries.

Now assume $d < \ell/2$.
By \cref{lem:ran_red}, taking $q = 4 d\ceil{\log(M \ell/d)} + \ceil{\log(1/\delta)}$ and computing $AZ \bmod M$ for a random $\ell$-by-$q$ 
Boolean matrix $Z$ allows us to learn $A$ with probability at least $1-\delta$.  By \cref{lem:Ay}
we can compute $AZ_i \bmod M$ with $\ceil{\log(M)}$ matrix cut queries, where $Z_i$ the $i^{th}$ column of $Z$.  Thus we can learn $A$ after $q \ceil{\log(M)}$ many 
matrix cut queries.
\end{proof}

\begin{definition}[degree sequence]
Let $A \in [M]^{k \times \ell}$ be a matrix.  The degree sequence of $A$ is the vector $x \in \N^k$ such that 
$x_i$ is the number of nonzero entries in the $i^{th}$ row of $A$, for all $i=1, \ldots, k$.  
\end{definition}

\begin{lemma}[Approximate degree sequence]
\label{lem:matrix_approx_degree}
Let $M,k,\ell$ be positive integers and $A \in [M]^{k \times \ell}$ be a matrix.  There is a quantum algorithm, making $O(\log(\ell M) (\log(\ell)+1) \log(\frac{k (\log(\ell) +1)}{\delta}))$ many 
matrix cut queries, that with probability at least $1-\delta$ outputs $\vec{g} \in \R^k$ that is a good estimate of the degree sequence of $A$.  If $M=2$ then there is a 
quantum algorithm that outputs the degree sequence of $A$ perfectly after $\ceil{\log(\ell+1)}$ many matrix cut queries.
\end{lemma}

\begin{proof}
Define $x^{(1)}, \ldots, x^{(k)} \in \{0,1\}^\ell$ by $x^{(i)}(j) = 1$ if $A(i,j) > 0$ and $x^{(i)}(j) = 0$ otherwise.  
The degree sequence of $A$ is then $(|x^{(1)}|, \ldots, |x^{(k)}|)$.  We apply
\cref{lem:app_count_analysis} to approximate $(|x^{(1)}|, \ldots, |x^{(k)}|)$ and therefore the degree sequence of 
$A$.  To do this, for a subset $S \subseteq \{1, \ldots, \ell\}$, we need to serve the $k$-OR query 
$(\OR(x^{(1)},S), \ldots, \OR(x^{(k)},S))$.  Let $\chi_S$ be the characteristic vector of $S$.  As the entries of $A$ 
are at most $M-1$ in magnitude, the entries of $A \chi_S$ are at most $\ell(M-1)$.  Therefore $A \chi_S = A \chi_S \bmod (\ell(M-1)+1)$.  
We can thus compute $A \chi_S$ with $\ceil{\log((\ell(M-1)+1)}$ many queries by \cref{lem:Ay}.  
Computing $A \chi_S$ suffices to serve the $k$-OR query since $(\OR(x^{(1)},S), \ldots, \OR(x^{(k)},S)) = 
(\min\{(A \chi_S)_1, 1\}, \ldots, \min\{(A \chi_S)_k, 1\})$.
Thus by \cref{lem:app_count_analysis} we can output a good estimate of the degree sequence of $A$ with probability at 
least $1-\delta$ after $O\left(\log(\ell M) (\log(\ell)+1) \log\left(\frac{k (\log(\ell)+1)}{\delta}\right)\right)$ many matrix cut queries.

For the case $M=2$, note that the degree sequence of $A$ is the vector $A \vec{1}$, where $\vec{1}$ is the vector of all ones.  
As the entries of $A\vec{1}$ are at most $\ell$ by \cref{lem:Ay} we can compute the degree sequence with $\ceil{\log(\ell+1)}$ many matrix 
cut queries.
\end{proof}

\begin{theorem}
\label{thm:matrix_m}
Let $A \in [M]^{k \times \ell}$ be a matrix with $m$ many nonzero entries.  There is a quantum algorithm to learn $A$ 
with probability at least $1-\delta$ after making 
\[
O\left(\sqrt{m \log(M \ell)}  \log(M)  + \log(\ell M) (\log(\ell)+1) \log\left(\frac{k(\log(\ell)+1)}{\delta}\right) \right)
\]
many matrix cut queries.
When $M=2$, a better bound of $O(\sqrt{m \log(\ell)} + \log(\ell+1) + \log(1/\delta))$ many matrix cut queries holds.
\end{theorem}

\begin{proof}
First we use \cref{lem:matrix_approx_degree} to compute $g \in \R^k$ which is a good approximation of the 
degree sequence of $A$ except with probability $\delta/2$.  This takes 
\[
O\left(\log(\ell M) (\log(\ell)+1) \log\left(\frac{k (\log(\ell)+1)}{\delta}\right)\right)
\]
many queries.
If $M=2$ then we can exactly compute the degree sequence of $A$ with $\ceil{\log(\ell+1)}$ many matrix cut queries.  
We now assume that $g$ is a good approximation and add $\delta/2$ to our total error bound.

Let $d \in \N$ be a degree threshold that will be chosen later, and define $L = \{i : g(i) \le d\}$ and $H = \{i : g(i) > d\}$.  As $g$ is a good 
approximation of the degree sequence, rows whose indices are in $L$ have at most $2d$ nonzero entries.  
We use \cref{lem:learn_matrix} to learn the submatrix $A(L, 1\colon \ell)$ with probability at least $1-\delta/2$ with 
$O( (d\log(M \ell/d) + \log(1/\delta))  \log(M))$ many matrix cut queries.  

For each $i \in H$, the $i^{th}$ row of $A$ must have at least $d/4$ many nonzero entries, as $g$ is a good approximation of the 
degree sequence of $A$.  Thus $|H| \le 4m/d$.  By \cref{cor:individual} we can learn the submatrix $A(H, 1 \colon \ell)$ with 
$4m \ceil{\log(M)}/d$ many matrix cut queries.

Setting $d = \sqrt{\frac{m}{\log(M\ell)}}$ the total number of queries becomes 
\[
O\left(\sqrt{m \log(M \ell)}  \log(M) + \log(\ell M) (\log(\ell)+1) \log\left(\frac{k(\log(\ell)+1)}{\delta}\right) \right) \enspace .
\]

In the case $M=2$, the number of queries becomes $O(\sqrt{m \log(\ell)} + \log(\ell+1) + \log(1/\delta))$.
\end{proof}

\section{Learning graphs} \label{sec:graph_learning}
A {\em weighted graph} is a couple
$G = (V,w)$, where $V$ is the set of vertices,  $V^{(2)}$ is the
set of subsets of $V$ with cardinality 2, and $w : V^{(2)} \rightarrow \N $ is the weight function. We assume that we have 
a total ordering $v_1 < v_2 < \cdots < v_{|V|}$ on the elements of $V$.  The set of edges
is defined as $E = \{e \in V^{(2)} : w(e) > 0\}$, therefore weighted graphs are undirected and without self-loops.
We can also think of them as multi-graphs, where the number of edges between vertices 
$u$ and $v$ is $w( \{u,v\})$. 
When the range of $w$ is $\{0,1\}$ we will call a weighted graph a {\em simple graph}, or just a {\em graph} and
denote it by $G = (V,E)$.  If $G$ is a weighted graph on $n$ vertices, the \emph{adjacency matrix} of $G$ is an $n$-by-$n$ matrix $A_G$ with 
zeros on the diagonal and $A_G(i,j) = w(\{v_i,v_j\})$, for $i \ne j$. 
Observe that $A_G$ is a symmetric matrix.

For an edge $e = \{u,v\} $, we say that $u$ and $v$ are the endpoints of $e$.
The degree $\deg(v)$ of a vertex $v$ is the number of edges for which $v$ is an endpoint.
For $S,T \subseteq V$ sets of vertices, we denote
by $E(S,T)$ the set of edges with one endpoint in $S$ and the other endpoint in $T$. (More precisely, $E(S,T) = \{e \in E: |e \cap (S \cup T) | = 2, |e \cap S| \ge 1, |e \cap T|\ge 1\}$.) We extend the weight function $w$
to sets of vertices $S,T \subseteq V$ by $w(S,T) = \sum_{e \in E(S,T)} w(e)$.

A {\em bipartite weighted graph} is a triple $G=(V_1, V_2, w)$, 
where $V_1 = \{v_1, \ldots , v_k\}$ and $V_2 = \{v_{k+1}, \ldots , v_{k + \ell}\}$ 
are the disjoint sets of respectively left and right vertices,
and $w : V_1 \times V_2 \rightarrow \N $ is the bipartite weight function.  A bipartite graph 
$G=(V_1, V_2, w)$ can of course also be viewed as a graph $G' =(V',w')$ with vertex set $V = V_1 \cup V_2$,
where $V_1$ and $V_2$ are independent sets, 
and weight function $w'$ over $V^{(2)} $, where 
\[
w' (\{v_i,v_j\}) = 
\begin{cases} 
0, & \mbox{ if } 1\leq i, j \leq k  \mbox{ or } k + 1 \leq i, j \leq k+ \ell, \\
w(v_{ \min \{i,j\}},v_{ \max \{i,j\}} ),  & \mbox{ otherwise }  .
\end{cases}
\]
Consistently with the general case, for $V'_1 \subseteq V_1$ and $V'_2 \subseteq V_2$, we extend $w$ as 
$w(V'_1, V'_2) = \sum_{u \in V'_1, v \in V'_2} w(u, v).$
The set of edges
is defined as $E = \{e \in V_1 \times V_2 : w(e) > 0\}$. 

The biadjacency matrix of $G$ is a $k \times \ell$ matrix $B_G$ where 
$B_G(i,j-k) = w(\{v_i,v_j\})$ for $1 \leq i \leq k, k+1 \leq j \leq k+ \ell$.  
Similarly to the bipartite weight function, the
biadjacency matrix $B_G$ is a condensed description of the $(k+\ell) \times (k+ \ell)$
adjacency matrix $A_G$ of $G$, where $A_G (i,j) = {w'} (\{v_i,v_j\})$, for $i \neq j$.

We will look at four different oracles for accessing a graph, the matrix cut oracle, the 
disjoint matrix cut oracle, the additive oracle, and the cut oracle. The comparison of their definitions and relative power is illustrated in \cref{fig:oracles}.

\begin{definition}[Matrix cut oracle and disjoint matrix cut oracle for a graph]
Let $G = (V,w)$ be a weighted graph. The matrix cut oracle  for $G$ is the matrix cut oracle for the
adjacency matrix $A_G$ of $G$. The disjoint matrix cut oracle for $G$ is the matrix cut oracle for $G$
with the restriction that it can only be  queried on strings
$x,y \in \{0,1\}^{|V|}$, where $x_i y_i = 0$, for all $1 \leq i \leq n$.
\end{definition}

If we consider the strings $x,y \in \{0,1\}^{|V|}$ as characteristic vectors of the sets $X,Y \subseteq V$,
the restriction on the domain of the disjoint matrix cut oracle is that it is only defined if $X \cap Y = \emptyset$, which explains its name.

Beside the matrix oracles we
will look at two more oracle models for accessing a graph, the additive oracle and the cut oracle.
\begin{definition}[Additive oracle]
Let $G = (V,w)$ be a weighted graph. The additive oracle $a : [2]^V \rightarrow \N$ returns $a(S) = \sum_{\{u,v\} \in S^{(2)}} w(\{u,v\})$ 
for any subset $S \subseteq V$.
\end{definition}
\begin{definition}[Cut oracle]
Let $G = (V,w)$ be a weighted graph. The cut oracle $c : [2]^V \rightarrow \N$ returns $c(S) = w(S, V \setminus S)$ 
for any subset $S \subseteq V$.
\end{definition}

Learning graphs with an additive oracle have been extensively studied in the classical case \cite{GK98,GK00, ABKRS04, Maz10,BM10,BM11}.  
For their proofs many of these papers actually work with matrix cut queries, and we find it useful to explicitly define these here.
The cut oracle has also been studied in the classical case in the context of  computing the minimum cut of a graph \cite{RSW18, MN20}.
We see in the next section that the cut oracle is the weakest of all these models, and our main algorithmic results will be for the 
cut oracle model.   
\begin{figure}[h]
\centering
\[
  \begin{tikzcd}[column sep = 1.5cm, cells = {anchor=west}]
    \mathrm{general} (S,T) \arrow[r,"1", shift left=0.5ex] \arrow[d,"1"] & \mathrm{additive} (S,S) \arrow[d,"3"] \arrow[l,"5",shift left=0.5ex]\\
     \mathrm{disjoint} (S,T) \arrow[r,"1",shift left=0.5ex] \arrow[u,"\Theta(n)",shift left=1.5ex]& \makebox[\widthof{$\mathrm{additive} (S,S)$}]{$\mathrm{cut} (S, \bar S)$} \arrow[u, "\Theta(n)", shift left=1.5ex] \arrow[l,"3", shift left=0.5ex]
  \end{tikzcd}
\]
\caption{Illustration of oracles for comparison. The arrows indicate reductions, and $A \xrightarrow{s} B$ means that a query to oracle $B$ can be simulated by $s$ queries to oracle $A$.}
\label{fig:oracles}
\end{figure}
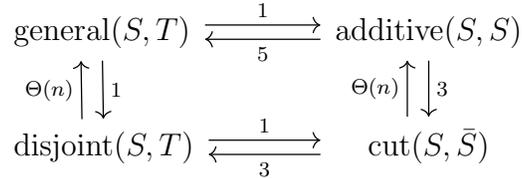


\subsection{Relationships between oracles}
In this section we examine the power of the four oracles we have defined for graphs.
In essence we show that if we consider the relative power of these oracles up to a constant overhead,
then the disjoint matrix cut oracle for graphs and the cut oracle have the same power,
the matrix cut oracle for graphs and the additive oracle have the same power, 
and the power of the latter group is greater than the power of the former.

\begin{definition}[Constant-reduction between oracles]
Let $O_1$ and $O_2$ be graph oracles.
We say that $O_1$ is constant-reducible to $O_2$ if there exist a positive integer $k$ such that
for every weighted graph $G = (V,w)$,
every query of $O_1$ to $G$ can be computed with $k$ queries of $O_2$ to $G$.
If $O_1$ and $O_2$ are mutually constant-reducible to each other, then they are called constant-equivalent.
\end{definition}

The first lemma shows the constant-equivalence of the disjoint matrix cut oracle for graphs and the cut oracle.
\begin{lemma}
\label{lem:3}
The disjoint matrix cut oracle for graphs and the cut oracle are constant-equivalent.
In particular, the cut oracle can simulate with $3$ queries a query of the disjoint matrix cut oracle for graphs.
\end{lemma}

\begin{proof}
The cut oracle is obviously 1-reducible to the disjoint matrix cut oracle. For the reverse direction,
let $G = (V,w)$ be a weighted graph with $|V| = n$ and let $A_G$ be its adjacency matrix.  
Let $x,y \in \{0,1\}^n$ be the characteristic vectors of $X,Y \subseteq V$ where $X$ and $Y$ are disjoint.
Then
\begin{align*}
x^T A_G y &= w(X,Y)  \\
                 &= \frac{1}{2}\left(c(X) + c(Y) - c(X \cup Y) \right) \enspace.
\end{align*}
\end{proof}

\begin{lemma}
\label{lem:biptograph}
Let $G$ be a bipartite graph. A matrix cut query to the biadjacency matrix of $G$ can be simulated
by $3$ cut queries.
\end{lemma}
\begin{proof}
Let $B_G$ be the biadjacency matrix of $G$ with $k$ left and $\ell$ right vertices,
 and let $x \in \{0,1\}^k, y \in \{0,1\}^\ell$.
We define the vectors $\bar x, \bar y \in \{0,1\}^{k+\ell}$ as
$\bar x = [x; 0^{\ell}]$ and $\bar y = [0^k; y]$.
Then $x^T B_G y = {\bar x}^T A_G \bar y$, where $A_G$ is the adjacency matrix of $G$.
Since $x$ and $y$ are the characteristic vectors of disjoint sets in $V_1 \cup V_2$,
by \cref{lem:3} we can compute $x^T B_G y$  with 3 queries to the cut oracle for $G$.
\end{proof}

The constant-equivalence of the matrix cut oracle for graphs and the additive oracle was essentially
shown by Grebinski and Kucherov in Theorem 4 of \cite{GK00}. For completeness we reproduce
here the proof.

\begin{lemma}
\label{lem:gk}
The matrix cut oracle for graphs and the additive oracle are constant-equivalent.
In particular, the latter oracle can simulate with $5$ queries a query of the former.
\end{lemma}

\begin{proof}
Let $G = (V,w)$ be a weighted graph on $n$-vertices, and let $A_G$ be its adjacency matrix.  
The additive oracle is 1-reducible to the matrix cut oracle for $A_G$ because for every
set $X \subseteq V$ with characteristic vector $x \in \{0,1\}^n$, we have $a(X) = \frac{1}{2} x^T A_G x$. 

For the reverse direction we show that the
matrix cut oracle to $A_G$ 
can be simulated with 5 queries to the additive oracle for $G$.
For this let $X,Y \subseteq V$ be arbitrary sets. We consider the characteristic vectors
$x_{-}, y_{-}$ and $z$ of respectively $X \setminus Y, Y \setminus X$ and $X \cap Y$.
Then 
\begin{align*}
x^T A_G y &= (x_{-} + z)^T A_G (y_{-} +z) \\
                  &=  x_{-}^T A_G y_{-} + x_{-}^T A_G z + z^T A_G y_{-} + z^T A_G z \\
                  &= \frac{1}{2}\left((x_{-}+y_{-})^TA_G (x_{-}+y_{-}) + x^T A_G x + y^T A_G y \right)
                   - x_{-}^T A_G x_{-} - y_{-}^T A_G y_{-} \enspace .
\end{align*}
To go from the second line to the last line we used  that $A_G$ is symmetric, 
that the sets $X \setminus Y, Y \setminus X$ and $X \cap Y$ are pairwise disjoint,
and finally the fact that if $u_1$ and $u_2$ are the characteristic vectors of the disjoint sets
$U_1, U_2 \subseteq V$ then 
$u_1^T A_G u_2 = a(U_1 \cup U_2) - a(U_1) - a(U_2)$.

 \end{proof}

Since the matrix cut oracle for graphs is by definition at least as strong as the disjoint matrix cut oracle
for graphs, \cref{lem:3} and \cref{lem:gk} imply that the cut oracle is constant reducible to
the additive oracle. The following lemma shows that the simulation in fact can be done by 3 queries.
\begin{lemma}
\label{lem:add_cut}
Let $G = (V,w)$ be a weighted graph.  The cut oracle to $G$ is $3$-reducible to the additive oracle for $G$.
\end{lemma}

\begin{proof}
For any $S \subseteq V$
\[
c(S) = a(V) - a(S) - a(V\setminus S) \enspace .
\]
\end{proof}

We now turn to the question of simulating an additive oracle with a cut oracle.
In an $n$-vertex weighted graph $G = (V,w)$, we can simulate an additive oracle with at most $3n$ applications 
of a cut oracle.  This is because for any $S \subseteq V$,
\[
a(S) = \frac{1}{2} \sum_{v \in S} w(v,S \setminus \{v\}) \enspace
\]
and each $w(v,S \setminus \{v\})$ can be computed with $3$ cut queries by \cref{lem:3}.  Note that this algorithm 
works no matter how large the weights are.  We next show that for weighted graphs with sufficiently large weights 
this trivial algorithm is nearly tight, and in the worst case $\Omega(n)$ cut queries are needed to simulate 
an additive query.

To do this, we first need a specific form of the Fredholm alternative.
\begin{lemma}[Fredholm Alternative]
\label{lem:fredholm}
Let $A \in \{0,1\}^{N \times k}$ have independent columns, and let $b \in \{0,1\}^N$. Suppose that $Ax = b$ has no solution $x \in \R^k$.  Then the 
integer vector $\hat y = \det(A^TA) (I - A(A^TA)^{-1}A^T) b \in \Z^N$ satisfies 
\begin{enumerate}
 \item $\hat y^TA = \vec{0} \; ,$
 \item $\hat y^Tb \ne 0 \; ,$ and
 \item $\|\hat y\|_\infty \le N^{k+1/2} k^{k/2}.$ 
\end{enumerate}
\end{lemma}

\begin{proof}
As $A$ has independent columns, $A^TA$ is invertible.
Let $b_c = A(A^TA)^{-1}A^T b$ be the orthogonal projection of $b$ onto the column space
of $A$ and $y = (I - A(A^TA)^{-1}A^T) b$ be the orthogonal projection onto the left nullspace.
Then $b = b_c + y$ and, as $A x = b$ has no solution, $y \ne \vec{0}$.  As $y$ is in the left 
nullspace, we have $y^T A = \vec{0}$.  Also $y^T b = y^T (b_c +y) = \|y\|^2 \ne 0$ as $y^T b_c =0$.

Now $\|y\| \le \|b\| \le \sqrt{N}$, thus also $\|y\|_\infty \le \sqrt{N}$.  For an invertible matrix $B$, by Cramer's rule $B^{-1} = \mathrm{adj}(B)/\det(B)$, 
where $\mathrm{adj}(B)$ is the adjugate matrix of $B$.
Thus $\det(A^TA) \cdot (A^TA)^{-1}=\mathrm{adj}(A^TA)$ is an integer matrix and $\hat y = \det(A^TA) y$ is an 
integer vector. As $\hat y$ is just a nonzero scalar factor of $y$, we have $\hat y^T A = 0$ and $\hat y^T b\ne 0$. Finally, $\| \hat y\|_\infty \le \sqrt{N}\cdot \det(A^TA) \le N^{k+1/2} k^{k/2}$, as the entries of $A^TA$ are 
at most $N$ and Hadamard's inequality gives that $\det(B) \le N^k k^{k/2}$ for a $k$-by-$k$ matrix $B$
whose entries are bounded by $N$.
\end{proof}

\begin{theorem}
\label{lem:cut_lower}
Any deterministic cut query algorithm on weighted graphs with vertex set $V$
must make at least $|V|/2$ queries to compute $a(V)$.
\end{theorem}

\begin{proof}
Let $G$ be a weighted graph with vertex set $V$ where $|V| = n$.
We will show that at least $n/2$ cut queries are needed to compute 
the total edge weight $a(V)$ in $G$.

Let $\symmat_n^0$ be the vector space of 
symmetric $n$-by-$n$ matrices with zeros along the diagonal. Let  the map 
$\symvec :  \symmat_n^0 \to \R^{\binom{n}{2}} $ be 
defined by $\symvec(C) = [C(2\upto n,1);C(3\upto n,2); \cdots ;C(n,n-1)]$, which
is an isomorphism between the two vector spaces.

Let $B = J - I$ where $J$ is the $n$-by-$n$ all ones matrix and $I$ is the $n$-by-$n$ identity matrix. 
Let $b = \symvec(B)$.
Then $a(V) = \symvec(A_G)^T b$ where $A_G$ is the adjacency matrix of $G$.
For a subset $X \subseteq V$ of the vertices, let $x \in \{0,1\}^n$ be the characteristic vector of $X$,
and let $\bar x \in \{0,1\}^n$ be the characteristic vector of $V \setminus X$. Then for the cut value of $X$
we have $c(X) =  \symvec(A_G)^T \symvec(x {\bar x}^T + \bar x x^T).$ Observe that the rank of
the matrix $x {\bar x}^T + \bar x x^T$ is 2.

Consider a deterministic cut query algorithm making $k < n/2$ many queries.
Let $m = n^{2k+1} k^{k/2}$, 
and let $G_1$ be the graph whose adjacency matrix is $A_{G_1} = m\cdot B$.
We set $a_1 = \symvec (A_{G_1})$.
Suppose that the algorithm makes queries $X_1, \ldots , X_k \subseteq V$ in $G_1$, with respective
characteristic sequences $x_1, \ldots, x_k \in \{0,1\}^n$.
Let $d_i = \symvec(x_i {\bar x_i}^T + \bar x_i x_i^T)$ for $i=1, \ldots, k$, and let $D$ be the 
$\binom{n}{2}$-by-$k$ matrix whose $i$th column is $d_i$.  

As the rank of $B$ is $n$ and $k < n/2$, 
there is no solution $x$ to the equation $Dx = b$, since otherwise we could express 
$B$ as a linear combination of $k$ matrices of rank 2.
Let $\hat y \in \Z^{\binom{n}{2}}$ be the vector given by \cref{lem:fredholm} 
satisfying $\hat y^T D = \vec{0}$ and $\hat y^T b \ne 0$.  
Then by  the choice of $m$ and by \cref{lem:fredholm}, $a_1 + \hat y$
will be a non-negative integer vector. Let $G_2$ be the weighted graph whose adjacency matrix satisfies
$ \symvec (A_{G_2}) = a_1 + \hat y$.
 
Now $a_1^T D = (a_1 + \hat y)^T D$ means that the graphs $G_1$ and $G_2$ give the same answers on 
all queries.  Yet $a_1^T b \ne (a_1 + \hat y)^T b$, meaning that the total sum of weights 
in $G_1$ and in $G_2$ are different.  
\end{proof}

\subsection{Learning graphs with an additive oracle}
\cref{lem:gk} immediately lets us apply \cref{lem:learn_matrix} and \cref{thm:matrix_m} to learning graphs with 
additive queries.
\begin{corollary}
\label{cor:add_learn}
Let $G=(V,w)$ be a weighted graph on $n >1$ vertices.
\begin{enumerate}
\item If every vertex has degree at most $d$ then there is is a quantum algorithm that 
learns $A$ with probability at least $1-\delta$ after making 
{$O( (d \log (Mn/d) + \log(1/\delta) )\log (M))$} 
many additive oracle queries to $G$.
\item If $G$ has $m$ edges then there is a quantum algorithm that learns $G$ with probability at least $1-\delta$ 
after making
\[
O\left(\sqrt{m  \log(M n)} \log(M)  + \log(Mn) \log(n) \log\left(\frac{n\log(n)}{\delta}\right) 
\right)
\]
many additive oracle queries to $G$.  When $M=2$, a better bound of $O(\sqrt{m \log(n)} + \log(n) + \log(1/\delta))$ many additive queries holds.
\end{enumerate}
\end{corollary}
Note that this result implies that one can learn a matching or Hamiltonian cycle using $O(\log n)$ quantum additive queries, in contrast to the
 $\Theta(n)$ additive queries needed deterministically \cite{GK00}.

\subsection{Learning graphs with a cut oracle}
Since by \cref{lem:cut_lower} a cut oracle cannot efficiently simulate an additive oracle (or a matrix cut oracle), we 
must do more work to adapt the results from \cref{sec:matrix_learning} to the 
case of a cut oracle.  We first show how to learn a general weighted graph with weights bounded by $M$ with $O(n \log(M))$
many cut queries.
\begin{theorem}
\label{thm:learning_ind}
Let $G = (V, w)$ be a weighted graph with $n$ vertices and edge weights in $[M]$.  There is a quantum algorithm that learns $G$ 
perfectly after making $3n \ceil{\log(M)}$ cut queries.  In particular, a simple graph can be learned perfectly by a quantum algorithm 
making $3n$ cut queries.
\end{theorem}

\begin{proof}
Let $A_G$ be the adjacency matrix of $G$, and let $v_1, \ldots, v_n$ be the order of the vertices labelling the rows and columns.  
Let $B_i = A_G(i, \{1, \ldots, i-1\} \cup \{i+1, \ldots, n\})$ be the 
$i^{th}$ row of $A_G$ with the $i^{th}$ entry (which must be $0$) removed.  Then $x^T B_i y$ 
can be computed with 3 cut queries by \cref{lem:biptograph} for any $x \in \{0,1\}, y \in \{0,1\}^{n-1}$.  Thus 
by \cref{cor:individual} we can learn $B_i$ with $3\ceil{\log M}$ many cut queries.  Doing this in turn for each $v_i$
gives the result.

The ``in particular'' statement follows by taking $M=2$.  
\end{proof}

Next we show how to learn low-degree and sparse graphs more efficiently.
We start out in the bipartite case. 

\begin{lemma}
\label{lem:bipartite_learning_degree}
Let $G=(V_1,V_2, w)$ be a weighted bipartite graph with $|V_1|=k, |V_2|=\ell$, 
and edge weights in $[M]$. 
\begin{enumerate}
\item If $\deg(u) \leq d$ for every $u \in L$ then
there is a quantum algorithm that learns $G$ with probability at least $1-\delta$ after making 
{$O( (d \log (M\ell/d) + \log(1/\delta)) \log M )$} many cut queries.
\item If $G$ has $m$ edges then there is a quantum algorithm that learns $G$ with probability at least $1-\delta$ 
after making
\[
O\left(\sqrt{m  \log(M \ell)} \log(M)  + \log(\ell M) (\log(\ell)+1) \log\left(\frac{k(\log(\ell)+1)}{\delta}\right) 
\right)
\]
many 
cut queries.
When $M=2$, a better bound of $O(\sqrt{m \log(\ell)} + \log(\ell+1) + \log(1/\delta))$ many cut queries holds.
\end{enumerate}
\end{lemma}

\begin{proof}
By \cref{lem:biptograph} any matrix cut query to the biadjacency matrix of $G$ can 
be computed by 3 queries to the cut oracle for $G$.
Item~(1) therefore follows from \cref{lem:learn_matrix}
and item~(2) follows from \cref{thm:matrix_m}.
\end{proof}

Now we extend the bipartite algorithms to the general case.

\begin{theorem}
\label{lem:cut_learn}
Let $G=(V,w)$ be a weighted $n$-vertex graph, $n > 1$, with edge weights in $[M]$. 
\begin{enumerate}
\item If $\deg(v) \leq d$ for every $v \in V$ then
there is a quantum algorithm that learns $G$ with probability at least $1-\delta$ after making 
$ O((d \log (Mn/d) + \log(\log(n)/\delta)) \log (M) \log(n) )$  many cut queries.
\item If $G$ has $m$ edges then there is a quantum algorithm that learns $G$ with probability at least $1-\delta$ 
after making
\[
O\left(\left(\sqrt{m  \log(M n)} \log(M)  + \log(Mn) (\log(n)) \log\left(\frac{n \log(n)}{\delta}\right)
\right) \log(n) \right)
\]
many cut oracle queries to $G$.  When $M=2$, a better bound of $O( (\sqrt{m \log(n)} + \log(n) + \log(\log(n)/\delta)) \log(n) )$ many cut queries holds.
\end{enumerate}
\end{theorem}

\begin{proof}
Let $r=\ceil{\log(n)}$.  We suppose that the vertices are labeled with $r$-bit
binary strings. We define $r$ many weighted bipartite graphs  $G_i = (L_i, R_i, w_i)$, for $i = 1, \ldots , r$, derived from $G$ as follows. 
$L_i$ consists of the vertices
of $V$ whose $i^{th}$ bit is 0, whereas $R_i$ consists of the vertices whose $i^{th}$ bit is 1. 
The weight function $w_i$ in $G_i$ is defined as 
\[
w_i (u, v) = 
\begin{cases} 
w(u, v) & \mbox{ if } (u, v) \in L_i \times R_i \\
0          & \mbox{ otherwise} \enspace .
\end{cases}
\]

Let $B_{G_i}$ be the biadjacency matrix of $G_i$.  Let $x \in \{0,1\}^{|L_i|}, y \in \{0,1\}^{|R_i|}$ and $X = \{v \in V: x_i(v) = 1\}, Y = \{v \in V: y_i(v) = 1\}$, and let $\chi_X$ and $\chi_Y$ be the characteristic sequence
of respectively $X$ and $Y$ in $V$.  
Then $x^T B_{G_i} y = \chi_X^T A_{G_i} \chi_Y$, and since $X$ and $Y$ are disjoint sets in $V$
we can compute it 
with $3$ cut queries to $G$ by \cref{lem:3}.  

We first prove item~(1).  As $G$ has degree at most $d$, each $B_i$ will also have degree at most $d$.  Therefore by applying the first item of
\cref{lem:bipartite_learning_degree} with error bound $\delta/r$ we can learn $B_i$ with 
$O( (d \log (Mn/d) + \log(\log(n)/\delta)) \log (M))$ 
many cut queries.  Doing this for each $B_i$ in turn gives item~(1).

Similarly for item~(2), as $G$ has at most $m$ edges so will each $B_i$.  Learning each $B_i$ in 
turn by the second item of 
\cref{lem:bipartite_learning_degree} with error bound $\delta/r$ gives the bound in item~(2).
\end{proof}

\subsection{Lower bounds}
\label{sec:lb}
In the classical case, an $\Omega(n)$ lower bound is known on the number of cut queries needed by a deterministic algorithm to determine if a graph 
is connected.  As observed by Harvey (Theorem~5.9 in \cite{Harvey08}), 
an $\Omega(n)$ cut query lower bound for connectivity follows from the deterministic communication complexity lower bound of $\Omega(n \log n)$ 
for connectivity by \cite{HMT88}, and the fact that the answers to cut queries can be communicated with $O(\log n)$ bits.  Clearly this lower bound applies to all 4 of the graph oracles 
we have discussed, as all of them have answers that can be communicated with $O(\log(n))$ bits.  For randomized communication protocols a $\Omega(n)$ lower bound 
is known for connectivity \cite{BFS86}, giving an $\Omega(n/\log(n))$ lower bound on the number of cut queries needed to solve connectivity.

The quantum case behaves differently.  There is an analogous communication complexity result: \cite{IKLSW12} show that the 
bounded-error quantum communication complexity of connectivity is $\Omega(n)$.  However, the standard way to turn a query 
algorithm into a communication protocol, developed by Buhrman, Cleve, and Wigderson (BCW) \cite{BCW98}, involves 
Alice and Bob sending back and forth the query and answer state.  In the case of cut queries, this is
a state of the form $\sum_{S \subseteq \{0,1\}^n} \alpha_S \ket{S}\ket{b_S}$ that is $n+ \log(n)$-qubits.  Thus a $k$-cut query quantum algorithm for connectivity 
gives a $kn$ qubit communication protocol for connectivity via the BCW simulation.  
Our polylogarithmic cut query algorithm for connectivity shows an example of a (partial) function where this blow-up from the BCW simulation is nearly optimal.  
We note that Chakraborty et al.\  \cite{CCMP19} recently gave the first example of a \emph{total} function showing that the BCW blow-up can be necessary.

We now show that our quantum algorithms for learning a general graph with $O(n)$ cut queries \cref{thm:learning_ind}, a graph of maximum degree $d$ with 
$\widetilde O(d)$ cut queries, and an $m$-edge graph with $\widetilde O(\sqrt{m})$ cut queries \cref{lem:cut_learn}
are all tight, even if the algorithms are equipped with the stronger additive oracle.

Recall that the inner product function $\IP_{N} : \{0,1\}^N \times \{0,1\}^N \rightarrow \{0,1\}$ on strings of length 
$N$ is defined as $\IP_{N}(x,y) = \PARITY(x \wedge y)$.  The bounded-error quantum communication 
complexity of $\IP_{N}$ is $\Omega(N)$ \cite{Kre}.

\begin{lemma}
\label{clm:ip}
Suppose there is a bounded-error quantum query algorithm that learns a graph on $n$ vertices with $k$ many additive queries.  Then there is 
a bounded-error quantum communication protocol for $\IP_{\binom{n}{2}}$ with $O(kn)$ qubits of communication.
\end{lemma}

\begin{proof}
We consider $[n]^{(2)}$, the set of possible edges in a graph with vertex set $[n]$.
Alice and Bob first agree on a bijective mapping $f: [\binom{n}{2}] \rightarrow [n]^{(2)}$.  Now say that Alice 
receives $x \in \{0,1\}^{\binom{n}{2}}$ and Bob receives $y \in \{0,1\}^{\binom{n}{2}}$, and they wish 
to compute $IP_{\binom{n}{2}}(x,y)$.  Bob first creates a graph $G_y$ with
vertex set $[n]$  and edge set 
$E = \{f(e): y_e = 1\}$.  Alice now runs the quantum query algorithm to learn $G_y$ by sending the query states 
to Bob, who serves each query and sends the state back to Alice.  As the query state is $n + \log(n)$ 
qubits, this will take communication $O(kn)$ if there are $k$ queries.  
\end{proof}

\begin{theorem}
\label{thm:lower}
Any quantum algorithm to learn an $n$-vertex graph with bounded-error must make $\Omega(n)$ many additive queries.
\end{theorem}

\begin{proof}
The theorem is immediate from \cref{clm:ip} and Kremer's $\Omega (n^2)$ lower bound on the quantum communication complexity of the inner product on $\binom{n}{2} $ bits.
\end{proof}

\begin{corollary}
\label{cor:lower_m}
For any $n$ and $m \le \binom{n}{2}$ there is a family of $n$-vertex graphs $\mathcal{G}_{n,m}$ with at most $m$ edges 
such that any quantum algorithm requires $\Omega(\sqrt{m})$ many additive queries to learn a graph from $\mathcal{G}_{n,m}$ 
with bounded-error.
\end{corollary}

\begin{proof}
Let $\mathcal{G}_{n,m}$ be the family of graphs that are arbitrary among the first $k$ vertices, where $k$ is the smallest 
integer such that $\binom{k}{2} \ge m$, and all remaining $n-k$ vertices are isolated.  By \cref{thm:lower} $\Omega(k) = \Omega(\sqrt{m})$ 
many additive queries are needed to learn a graph from this family.
\end{proof}

\begin{corollary}
\label{cor:lower_d}
For any $n$ and $d \le n$ there is a family of $n$-vertex graphs $\mathcal{G}_{n,d}$ with each vertex having degree at most $d$ such that any quantum algorithm requires $\Omega(d \log(n/d))$ many additive queries to learn a graph from $\mathcal{G}_{n,d}$ 
with bounded-error.
\end{corollary}

\begin{proof}
Let $\mathcal{G}_{n,d}$ be the set of all $d$ regular graphs on $n$-vertices.  It is known that $|\mathcal{G}_{n,d}| \ge \left(\frac{n}{2ed}\right)^{nd/2}$ (see the argument above Theorem 4 in \cite{GK00}, or the very general results in \cite{LW17}).
Letting $\ell = \floor{\frac{nd}{2} \log\left(\frac{n}{2ed}\right)}$, Alice and Bob can therefore agree on a bijection between $\{0,1\}^\ell$ and a subset of $\mathcal{G}_{n,d}$.  As in the proof of \cref{clm:ip}, they then can use an algorithm to learn graphs 
in $\mathcal{G}_{n,d}$ with $k$ additive queries to solve $\IP_\ell$ with quantum communication complexity $O(kn)$.  Thus $k = \Omega(\ell/n) = \Omega(d \log(n/d))$.  
\end{proof}

\section{A quantum algorithm for computing connected components with cut queries} \label{sec:connect}
In this section, we give a quantum algorithm that outputs the connected components of an $n$-vertex graph after 
making $O(\log(n)^6)$ cut queries.  By \cref{lem:add_cut}, this implies the same result with respect to additive queries.

We first give a high level overview of the algorithm.  Let $G$ be an $n$-vertex graph with vertex set $V$ and $A_G$ its adjacency matrix.  
For this high level overview, we will assume that we have matrix cut query access to $A_G$.  This case contains all of the conceptual ideas 
needed and eliminates several technical issues that arise in the cut query case.

The algorithm proceeds in rounds.  In every round, we maintain a partition of $V$.  In a general round, we represent the partition as 
$\{S_1, \ldots, S_k, C_1, \ldots, C_t\}$ and let $\Scal = \{S_1, \ldots, S_k\}$ and $\Ccal = \{C_1, \ldots, C_t\}$.  We refer to sets of vertices as 
supervertices.  The algorithm maintains 
the invariants that each $\Ccal_j$ is a connected component, i.e. the induced subgraph on $C_j$ is connected and $|E(C_j, V \setminus C_j)| = 0$, and 
that for each $S_i$ the induced subgraph on $S_i$ is connected.  As the sets $C_j$ are connected components, they require no further processing 
and thus in a round we focus only on the sets in $\Scal$.  We initially set $\Scal = \tilde V =  \{ \{v\} : v \in V\}$ and $\Ccal = \emptyset$. 

An important concept for the algorithm will be the adjacency matrix of the weighted graph among the sets of $\Scal$.  
For $\Scal = \{S_1, \ldots, S_k\}$ define a $k$-by-$k$ matrix $A_{\Scal}$ where $A_{\Scal}(i,j) = |E(S_i, S_j)|$.  A key fact is that matrix cut query access to 
$A_G$ allows us to implement matrix cut queries on $A_{\Scal}$.  For $x,y \in \{0,1\}^k$ define $u,v \in \{0,1\}^n$ to be the characteristic vectors of the sets
$\cup_{i: x_i = 1} S_i$ and $\cup_{i: y_i = 1} S_i$ respectively.  Then we have $x^T A_{\Scal} y = u^T A_G v$.  

A general round proceeds as follows.  We take a degree threshold $d = \Theta(\log(n)^2)$.  The superdegree of a supervertex $S_i$ is the number of 
$j \ne i$ such that $|E(S_i, S_j)| > 0$.  Supervertices with superdegree at most $d$ we call low superdegree, and supervertices with superdegree 
greater than $d$ we call high superdegree.
\begin{enumerate}
  \item Estimate the superdegree of all supervertices by \cref{lem:matrix_approx_degree}  ($O(\log(n)^3)$ queries).
  \item For low superegree supervertices learn all their neighbors $S_j$ by \cref{lem:learn_matrix} ($O(\log(n)^4)$ queries taking $d = \Theta(\log(n)^2)$).
  \item Randomly sample $\Rcal \subseteq \Scal$ satisfying $|\Rcal| \le |\Scal|/2$.  With high probability 
  for every high degree supervertex $S_i$ there is an $R \in \Rcal$ with $|E(S_i, R)| > 0$, and we again use \cref{lem:learn_matrix} to find 
  such an $R$ for every high degree $S_i$ with $O(\log(n)^4)$ queries.
  \item Merge all supervertices $S_i, S_j$ for which we know $|E(S_i, S_j)| > 0$ and update $\Scal$ and $\Ccal$ accordingly.  
\end{enumerate}
We show that with high probability after this merging step the number of supervertices in $\Scal$ is at most half of what it was at the beginning of the round.  
This follows because with high probability every high degree supervertex will be merged with an element of $\Rcal$.
For low degree supervertices, we learn all of their neighbors.  Therefore if the connected component of a low degree supervertex $S_i$ only contains 
low degree supervertices, this step will learn its entire connected component and we can add this component to $\Ccal$.  Otherwise there is some 
low degree supervertex $S_j$ in the connected component of $S_i$ that is connected to a high degree supervertex $S_\ell$.  We will learn this connection 
in Step~(2), and learn a neighbor of $S_\ell$ in $\Rcal$ in step~(3).  This means that in Step~(4) $S_i$ will be merged into a supervertex that also contains 
some element of $\Rcal$.  Thus every supervertex for which we have not already learned its connected component will be merged with a supervertex in $\Rcal$ 
with high probability
and at the end of the round the total number of supervertices in $\Scal$ is reduced by at least half.  The algorithm thus terminates after $O(\log n)$ rounds and the 
total complexity becomes $O(\log(n)^5)$ many matrix cut queries.  

We have seen that cut queries cannot efficiently simulate matrix cut queries in general.  
However, by \cref{lem:biptograph}, cut queries can simulate matrix cut queries on \emph{bipartite} 
weighted graphs.  The trick we use to adapt the above algorithm to the cut query case is to always work with bipartite graphs.  For a 
graph $G$ we associate $\ceil{\log n}$ many bipartite graphs $(L_i, R_i, E_i)$ for $i \in \ceil{\log n}$, where $L_i = \{v : v_i = 0\}, R_i = \{v : v_i = 1\}$ 
and $E_i = \{ (u,v) \in V \times V : \{u,v\} \in E, u \in L_i, v \in R_i\}$.  Note that every edge of $G$ appears as an edge in some $(L_i, R_i, E_i)$.  We essentially run steps 
(1)--(3) above on each $(L_i, R_i, E_i)$ separately, then incorporate all the information learned in the merge step in~(4).  Having to iterate over 
these $O(\log(n))$ many bipartite graphs results in an extra multiplicative logarithmic factor in the complexity, resulting in the claimed bound of 
$O(\log(n)^6)$ cut queries.

\subsection{Auxiliary subroutines}
In this subsection we go over some auxiliary subroutines that will be used in the cut query algorithm for connectivity.
We will make use of the following definition.
\begin{definition}[Supervertex, Superdegree]
A \emph{supervertex} is a subset of $S \subseteq V$ of vertices.  We say that a supervertex $S$ is connected if the subgraph induced on $S$ is connected.

We say that a set of supervertices $\Scal = \{S_1, \ldots, S_k\}$ is \emph{valid} if $S_i \cap S_j = \emptyset$ for all $i \ne j$.
We say that two valid sets of supervertices $\Scal = \{S_1, \ldots, S_k\}$ and $\Tcal =\{T_1, \ldots, T_\ell\}$ are \emph{disjoint} if 
$S_i \cap T_j = \emptyset$ for all $i \in \{1,\ldots, k\}, j \in \{1,\ldots, \ell\}$.

We say that there is a \emph{superedge} between supervertices $S_1, S_2$ if $|E(S_1, S_2)| > 0$.
Given two disjoint valid sets of supervertices $\Scal = \{S_1, \ldots, S_k\}$ and $\Tcal = \{T_1, \ldots, T_\ell\}$, the \emph{superdegree of $S_i$ into 
$\Tcal$}, denoted $\deg_{\Tcal}(S_i)$, is the number of $j \in [\ell]$ such that $|E(S_i, T_j)| >0$.
\end{definition}

\begin{definition}
Let $\Scal = \{S_1, \ldots, S_k\}$ and $\Tcal = \{T_1, \ldots, T_\ell\}$ be disjoint valid sets of supervertices of an $n$-vertex graph $G$.  
Define the \coisa between $\Scal$ and $\Tcal$ as the matrix $B \in [n^2]^{k \times \ell}$ where $B(i,j) = |E(S_i, T_j)|$.
\end{definition}

\begin{lemma}
\label{lem:B_query}
Let $\Scal = \{S_1, \ldots, S_k\}$ and $\Tcal = \{T_1, \ldots, T_\ell\}$ be disjoint valid sets of supervertices of an $n$-vertex graph $G$.  
Let $B$ be the \coisa between $\Scal$ and $\Tcal$.  A matrix cut query to $B$ can be answered with 
3 cut queries to $G$. 
\end{lemma}

\begin{proof}
For $x \in \{0,1\}^k, y \in \{0,1\}^\ell$ we have $x^T B y = |E(\cup_{i: x(i)=1} S_i, \cup_{j: y(j) = 1} T_j)|$ and thus can be computed with 
3 cut queries by \cref{lem:3}.
\end{proof}

\begin{corollary}[Approximate Degree Sequence]
\label{lem:ADS}
There is a quantum algorithm Approximate Degree Sequence$(\Scal, \Tcal, \delta)$ that takes as input disjoint valid sets of supervertices
$\Scal = \{S_1, \ldots, S_k\}$ and $\Tcal = \{T_1, \ldots, T_\ell\}$ of an $n$-vertex graph $G$ and an error parameter $\delta$, and with probability at least $1-\delta$ 
outputs $\vec{g} \in \R^k$ that is a good estimate of $(\deg_{\Tcal}(S_1), \ldots, \deg_{\Tcal}(S_k))$.
The number of cut queries made is $O(\log(\ell n)(\log(\ell) + 1) \log(k (\log(\ell)+1)/\delta))$.
\end{corollary}

\begin{proof}
Let $B$ be the \coisa between $\Scal$ and $\Tcal$.  By \cref{lem:B_query} a matrix cut query to $B$ can be answered by 
3 cut queries to $G$.  The result then follows from \cref{lem:matrix_approx_degree}.
\end{proof}

\begin{algorithm}
\caption{Approximate Degree Sequence$(\Scal,\Tcal,\delta)$}
\label{alg:approx_degree}
\hspace*{\algorithmicindent} \textbf{Input:} Disjoint valid sets of supervertices $\Scal = \{S_1, \ldots, S_k\}, \Tcal = \{T_1, \ldots, T_\ell\}$ 
and an error parameter $\delta$ \\
\hspace*{\algorithmicindent} \textbf{Output:} Vector $\vec{g} \in \mathbb{R}^k$ that is a good estimate of 
$(\deg_{\Tcal}(S_1), \ldots, \deg_{\Tcal}(S_k))$.
\end{algorithm}

\begin{algorithm}
\caption{\text{Learn Low}$(\Scal,\Tcal,h, \delta)$}
\label{alg:learn_low}
 \hspace*{\algorithmicindent} \textbf{Input:} Disjoint valid sets of supervertices $\Scal = \{S_1, \ldots, S_k\}, \Tcal =\{T_1, \ldots, T_\ell\}$, 
 a degree parameter $h$ such that $\deg_{\Tcal}(S_i) \le h$ for all $S_i \in S$, and an error parameter $\delta$. \\
 \hspace*{\algorithmicindent} \textbf{Output:} A $k$-by-$\ell$ matrix $B$ where $B(i,j) = |E(S_i, T_j)|$.
\end{algorithm}

\begin{corollary}[Learn Low]
\label{lem:learn_low}
There is a quantum algorithm Learn Low$(\Scal, \Tcal, h, \delta)$ that takes as input disjoint valid sets of supervertices
$\Scal = \{S_1, \ldots, S_k\}$ and $\Tcal = \{T_1, \ldots, T_\ell\}$ of an $n$-vertex graph $G$, a degree parameter $h$ such that 
$\deg_{\Tcal}(S_i) \le h$ for all $S_i \in \Scal$, and an error parameter $\delta$.  With probability at least $1-\delta$ 
Learn Low$(\Scal, \Tcal, h, \delta)$ outputs the \coisa between $\Scal$ and $\Tcal$.  The number of cut queries made is 
$O( (h \log(n /h) + \log(1/\delta)) \log (n) )$.
\end{corollary}

\begin{proof}
Let $B$ be the \coisa between $\Scal$ and $\Tcal$.  The rows of $B$ have at most $h$ nonzero entries by assumption and 
all entries are $O(n^2)$.  By \cref{lem:B_query} we can answer a matrix cut query to $B$ by 3 cut queries to $G$.  The result then follows 
from \cref{lem:learn_matrix}.  
\end{proof}

We accomplish Step~(3) in the high level description of the algorithm by a routine called Reduce High.  We first need a 
sampling lemma.
\begin{lemma}
\label{lem:sample_super}
Let $\Scal=\{S_1, \ldots, S_k\}$ and $\Tcal = \{T_1, \ldots, T_\ell\}$ be disjoint valid sets of supervertices of an $n$-vertex graph $G$, 
and suppose $t/8 \le \deg_{\Tcal}(S_i) \le 2t$.  Randomly sample 
with replacement $\frac{16\ell \ln(kn)}{t}$ many supervertices in $\Tcal$, and call the resulting set $\Rcal$. Then, except with 
probability $O(n^{-2})$, the superdegree of each $S_i$ into $\Rcal$ will be at least $1$ and at most $192 \ln(n)$.
\end{lemma}

\begin{proof}
Define $x^{(i)} \in \{0,1\}^\ell$ for $i \in \{1, \ldots, k\}$ by $x^{(i)}(j) = 1$ iff $|E(S_i, T_j)| > 0$.  Then the lemma follows by applying 
\cref{lem:sample} to $x^{(1)}, \ldots, x^{(k)}$ with $\delta = 1/n^2$ and using the fact that $k \le n$.  
\end{proof}

\begin{lemma}[Reduce High]
\label{lem:reduce_high}
The quantum algorithm Reduce High$(\Scal,\Tcal,d,\vec{g})$ (\cref{alg:reduce_high}) takes as input disjoint valid sets of supervertices
$\Scal = \{S_1, \ldots, S_k\}$ and $\Tcal = \{T_1, \ldots, T_\ell\}$ of an $n$-vertex graph $G$, a degree parameter $d$ such that 
$\deg_{\Tcal}(S_i) \ge d$ for all $S_i \in \Scal$, and a vector $\vec g \in \R^k$ that is a good estimate of $(\deg_{\Tcal}(S_1), \ldots, \deg_{\Tcal}(S_k))$.
\text{Reduce High} $(\Scal,\Tcal,d,\vec{g})$ makes $O(\log(n)^4)$ many cut queries and with probability at least 
$1-O(\log(n)/n^2)$ outputs a $k$-by-$\ell$ Boolean matrix $B$ such that 
\begin{enumerate}
\item Every row of $B$ has at least one nonzero entry.
\item If $B(i,j) \ne 0$ then $B(i,j) = |E(S_i, T_j)|$.
\item $B$ has at most $\frac{256 \ell \ln(n)}{d}$ many nonzero columns.
\end{enumerate}  
\end{lemma}

\begin{proof}
The algorithm is given in \cref{alg:reduce_high}.
First we bound the number of queries made.  We do $O(\log n)$ iterations of the 
for loop, and within each iteration queries only occur in Line~\ref{line:ll}, which is a call Learn Low with an $O(\log n)$ degree parameter.  Each such 
call to Learn Low requires $O(\log(n)^3)$ many cut queries by \cref{lem:learn_low}
Thus overall the number of queries is $O(\log(n)^4)$.

Now we show the correctness.  As $\vec{g} \in \R^k$ provides good estimates and
$\deg_{\Tcal}(S_i) \ge d$ for all $S_i \in \Scal$, we have $g(i) \ge d/2$ for all $i \in \{1,\ldots, k\}$.  Also $g(i) \le 4 \ell$, 
thus each $S_i$ will be put into $\Hcal_q$ for some value of $q$ in the loop.  

Let us now consider a particular iteration of the for loop when $q = j$.  Let $\Hcal_j = \{S_i \in \Scal: 2^{j-1} < g(i) \le 2^j \}$.  
As $\vec{g}$ contains good estimates, this means $2^{j-3} < \deg_{\Tcal}(S_i) \le 2^{j+1}$ 
for all $S_i \in \Hcal_j$.  
We are thus in the setting of \cref{lem:sample} with $t = 2^j$.  As $\Rcal_j$ is a $\frac{16 \ell \ln(|H_j| n)}{2^j}$-sample 
from $\Tcal$ the conclusion of \cref{lem:sample} gives that the superdegree of every vertex in $\Hcal_j$ into $\Rcal_j$ is 
between $1$ and $192 \ln(n)$, except with probability $O(n^{-2})$.   We now assume we are in this 
good case.  Then the upper bound on the degree passed to Learn Low is valid, 
and Learn Low will return the weighted biadjacency matrix between $\Hcal_j$ and $\Rcal_j$
with probability $1 - O(n^{-2})$ by \cref{lem:learn_low}

As every $S_i$ is in $\Hcal_q$ for one call of the for loop, each row of $B$ will have at least one nonzero entry.  Further, as each call to 
Learn Low returns a correct weighted biadjacency matrix between $\Hcal_q$ and $\Rcal_q$ 
except with probability $1/n^2$, $B$ will satisfy item~(2) except with probability at most $O(\log(n)/n^2)$.

Finally, the only columns of $B$ that can be nonzero are those 
indexed by sets that appeared in $\Rcal_q$ at some point of the algorithm.  As 
\begin{align*}
|\cup_q \Rcal_q| &\le 32 \ell \ln(n) \cdot \left( \sum_{q=\floor{\log d}-1}^{\ceil{\log \ell}+2} \frac{1}{2^q} \right) \\
& \le 32 \ell \ln(n) \cdot \frac{4}{d} \sum_{j=0}^{\infty} \frac{1}{2^j} \\
& \le \frac{256 \ell \ln(n)}{d} \enspace,
\end{align*}
the total number of nonzero columns of $B$ is at most $\frac{256 \ell \ln(n)}{d}$.
\end{proof}

\begin{algorithm}
\caption{\text{Reduce High}$(\Scal,\Tcal,d,\vec{g})$}
\label{alg:reduce_high}
 \hspace*{\algorithmicindent} \textbf{Input:} Disjoint valid sets of supervertices 
 $\Scal = \{S_1, \ldots, S_k\}, \Tcal = \{T_1, \ldots, T_\ell\}$ of an $n$-vertex graph $G$, a degree parameter $d$ such that 
 $\deg_{\Tcal}(S_i) \ge d$ for all $S_i \in \Scal$, and a vector 
 $\vec{g} \in \R^k$ such that $\vec{g}(i)/4 \le \deg_{\Tcal}(S_i) \le 2 \vec{g}(i)$ for all $i \in \{1, \ldots, k\}$. \\
  \hspace*{\algorithmicindent} \textbf{Output:} A $k$-by-$\ell$ matrix $B$ satisfying the conditions of \cref{lem:reduce_high}.
\begin{algorithmic}[1]
\State $B \gets \zeros(k,\ell)$
\For{$q = \floor{\log d}-1$ to $\ceil{\log \ell}+2$}
  \State $\Hcal_q = \{ S_i \in \Scal : 2^{q-1} < \vec{g}(i) \le 2^{q}\}$ 
  \State $\Rcal_q \gets$ Randomly choose $\frac{16 \ell \ln(|\Hcal_q| n)}{2^q}$ supervertices in $\Tcal$, with replacement 
  \State $B(\ind(\Hcal_q), \ind(\Rcal_q))  \gets \text{Learn Low}(\Hcal_q,\Rcal_q, 192 \ln(n), 1/n^2)$ \label{line:ll}
  \EndFor
\State \Return $B$
\end{algorithmic}
\end{algorithm}

\begin{algorithm}
\caption{Contract$(\Scal,\mathcal{A}, \text{low})$}
\label{alg:contract}
 \hspace*{\algorithmicindent} \textbf{Input:} Valid set of connected supervertices $\Scal = \{S_1, \ldots, S_k\}$, a list of $k$-by-$k$ weighted adjacency matrices 
 $\mathcal{A}=(A_1, \ldots, A_m)$ with rows and columns labeled by elements of $\Scal$, a vector $\text{low} \in \{0,1\}^k$ indicating if each set $S_i$ is low. \\
 \hspace*{\algorithmicindent} \textbf{Output:} Sets of supervertices $\Scal', \Ccal$, where each supervertex in $\Scal', \Ccal$ is connected, and moreover the 
 supervertices in $\Ccal$ are connected components.
\begin{algorithmic}[1]
\State $L \gets [ (S_i, \text{low}(i) ) : S_i \in \Scal]$
\For{$A \in \Acal$}
  \For{$(i,j) \in [k]^{(2)}$}
    \If{$A(i,j) > 0$}
      \State Pop $(U,\text{flag1}) \in L$ such that $S_i \subseteq U$
      \State Pop $(W,\text{flag2}) \in L$ such that $S_j\subseteq W$
      \State $U \gets U \cup W$
      \State $\text{lowFlag} = \text{flag1} \wedge \text{flag2}$
      \State Append $(U, \text{lowFlag})$ to $L$
    \EndIf
  \EndFor
\EndFor
\State $\Scal' \gets \emptyset$
\State $\Ccal \gets \emptyset$
\For{$(U, \text{lowFlag}) \in L$}
  \If{\text{lowFlag} = 0}
    \State $\Scal' \gets \Scal' \cup \{U\}$
  \Else
    \State $\Ccal \gets \Ccal \cup \{U\}$ \label{line:ccal}
  \EndIf
\EndFor
\State \Return $\Scal',\Ccal$  
\end{algorithmic}
\end{algorithm}

\begin{algorithm}
\caption{Shrink$(\Scal,d)$}
\label{alg:shrink}
 \hspace*{\algorithmicindent} \textbf{Input:} A valid set of connected supervertices $\Scal =\{S_1, \ldots, S_k\}$, and a degree parameter $d$. \\
 \hspace*{\algorithmicindent} \textbf{Output:} A set of connected supervertices $\Scal'$, and a set $\Ccal$ of connected components.
\begin{algorithmic}[1]
\State $\text{low}  \gets \ones(k,1)$
\State $\Acal \gets [\;]$
\For{$j = 1$ to $\ceil{\log(k)}$}
  \For{$ b \in \{0,1\}$}
  \State $\Lcal_{j,b} = \{S_t \in \Scal : t_j = b\}$
  \State $\Rcal_{j,b} = \{S_t \in \Scal : t_j = 1-b\}$
  \State $\vec{g} \gets$ Approximate Degree Sequence$(\Lcal_{j,b}, \Rcal_{j,b}, 1/n)$
  \State $\Hcal \gets \{S_t \in \Lcal_{j,b} : \vec{g}(t) \ge d\}$ 
  \State $\text{low}(\ind(\Hcal)) = 0$
  \State $B_{j,b} \gets \zeros(k,k)$
  \State $B_{j,b}(\ind(\Hcal),\ind(\Rcal_{j,b})) \gets \text{Reduce High}(\Hcal, \Rcal_{j,b}, d/4, \vec{g}(\ind(\Hcal)))$ \label{line:reduce_high}
  \State $\Lcal \gets \{S_t \in \Lcal_{j,b} : \vec{g}(t) < d\}$
  \State $C_{j,b} \gets \zeros(k,k)$
  \State $C_{j,b}(\ind(\Lcal), \ind(\Rcal_{j,b})) \gets \text{Learn Low}(\Lcal, \Rcal_{j,b}, 2d, 1/n)$   
  \State Append $B_{j,b}, C_{j,b}$ to $\Acal$
  \EndFor
\EndFor
  \State $(\Scal',\Ccal) \gets$ Contract$(\Scal, \Acal, \text{low})$ \label{line:contract}
  \State \Return $(\Scal',\Ccal)$
\end{algorithmic}
\end{algorithm}

\begin{algorithm}
\caption{Connectivity algorithm with cut queries}
\label{alg:q_con}
 \hspace*{\algorithmicindent} \textbf{Input:} Cut oracle for a graph $G$ on $n$ vertices \\
 \hspace*{\algorithmicindent} \textbf{Output:} Connected components of $G$
\begin{algorithmic}[1]
\State $\Scal \gets  \tilde V$  
\State $\text{ConComp} \gets \emptyset$
\Repeat
  \State $(\Scal, \Ccal) \gets Shrink(\Scal,1024 \ceil{\log n}^2)$
  \State $\text{ConComp} \gets \text{ConComp} \cup \Ccal$
\Until{$\Scal = \emptyset$}
\State \Return $\text{ConComp}$
\end{algorithmic}
\end{algorithm}

\begin{lemma}
\label{lem:contract}
Let $\Scal = \{S_1, \ldots, S_k\}$ be a valid set of supervertices, $\Acal$ a list of $k$-by-$k$ weighted adjacency matrices, 
and $\text{low} \in \{0,1\}^k$ a Boolean vector with the following properties:
\begin{enumerate}
  \item Every supervertex in $\Scal$ is connected.
  \item For every $A \in \mathcal{A}$ if $A(i,j) > 0$ there is a superedge between $S_i$ and $S_j$.
  \item For every $(i,j)$ such that $\text{low}(i) = 1$ and there is a superedge between $S_i$ and $S_j$,
  there is an $A \in \mathcal{A}$ with $A(i,j) > 0$.
\end{enumerate}
Then the algorithm Contract$(\Scal, \mathcal{A}, \text{low})$ given in \cref{alg:contract} outputs sets of supervertices $\Scal', \Ccal$ such that 
every $U \in \Scal'$ is connected, every $W \in \Ccal$ is a connected component, and $\Scal' \cup \Ccal$ is a partition of $\cup_{S \in \Scal} S$.  
Moreover, for every $U \in \Scal'$ there is an $S_i \subseteq U$ with $\text{low}(i) = 0$.
\end{lemma}

\begin{proof}
We first show that all supervertices in $\Scal', \Ccal$ are connected.  This follows because each $S_i \in \Scal$ is connected and 
we only merge two supervertices $U$ and $W$ when there is an $S_i \subseteq U, S_j \subseteq W$ and an $A \in \Acal$ with $A(i,j) > 0$.  
As by hypothesis  $A(i,j) > 0$ implies there is a superedge between $S_i$ and $S_j$ this means that $U$ and $W$ are in fact connected.

Next we show that every $W \in \Ccal$ is a connected component.  Suppose for a contradiction that this is not the case and therefore there 
is a $w \in S_i \subseteq W$ and $u \in S_j \subseteq (\cup_{t=1}^k S_t) \setminus W$ such that $\{u,w\}$ 
is an edge of $G$.  It must be the 
case that $S_i$ is low, as otherwise the lowFlag for $W$ would have been set to $0$ and $W$ would have been placed in $\Scal'$.  Thus 
$S_i$ must be low and therefore by hypothesis for some $A \in \Acal$ it is the case that $A(i,j) > 0$.  This means that at some point in 
Contract a set containing $S_i$ would have been merged with a set containing $S_j$, a contradiction to the fact that 
$S_j \subseteq (\cup_{t=1}^k S_t) \setminus W$.

The fact that $\Scal \cup \Ccal$ is a partition of $\cup_{S \in \Scal} S$ follows because at all times $\cup_{(U,flag) \in L} U$ is equal 
to $\cup_{S \in \Scal} S$.  This is true when $L$ is first defined, and is preserved when sets are popped from $L$, merged, and put 
back into $L$.

Finally, the ``moreover'' statement holds as if $\text{low}(i) = 1$ for all $S_i \in U$ then the lowFlag variable for $U$ will be set to $1$ and 
therefore $U$ will be placed into $\Ccal$ on Line~\ref{line:ccal}.
\end{proof}

\subsection{The shrink subroutine}
We now package Approximate Degree Sequence, Learn Low, Reduce High, and Contract together into our algorithm for finding the connected components 
of a graph.  

\begin{lemma}
\label{lem:shrink}
Let $\Scal$ be a valid set of connected supervertices and $d \in \mathbb{N}$ be a degree parameter given as input to \cref{alg:shrink}.  Then except with probability $O(\log(n)/n)$
the following two statements hold.
\begin{enumerate}
\item \cref{alg:shrink} outputs sets of supervertices $\Scal',\Ccal$ such that
\begin{enumerate}
  \item $|\Scal'| \le 512 \ceil{\log(|\Scal|)} \ln(n) |\Scal|/d$.  
  \item All supervertices in $\Scal'$ are connected.
  \item All supervertices in $\Ccal$ are connected components.  
  \item $\Scal' \cup \Ccal$ is a partition of $\cup_{S \in \Scal} S$.
\end{enumerate}
  \item The total number of cut queries made is $O(\log(n)^5 + d \log(n)^3)$.  
\end{enumerate}
\end{lemma}

\begin{proof}
All of Approximate Degree Sequence, Reduce High, and Learn Low have error probability at most 
$O(1/n)$.  As they are called at most $O(\log n)$ times, with probability at least $1-O(\log n/n)$ they will 
all return as promised.  We now argue correctness assuming this is the case.

Let us establish that the hypotheses of \cref{lem:contract} hold when Contract is called on Line~\ref{line:contract}.  
By assumption all supervertices in $\Scal$ are connected, thus Item~(1) holds.  Also, as we are in the case where Reduce High and Learn Low 
perform correctly, Item~(2) holds.  If $\text{low}(i) = 1$ in the call to Contract then for all $j,b$ for which 
$S_i \in \Lcal_{j,b}$ it holds that $\vec{g}(i) < d$, and thus the corresponding call to Learn Low learns all neighbors of $S_i$ in 
$\Rcal_{j,b}$.  As this is true for all $j,b$, we learn all neighbors of $S_i$, meaning that Item~(3) also holds.  As the hypotheses to 
\cref{lem:contract} hold, this means all supervertices in $\Scal'$ are connected, $\Ccal$ contains connected components, and 
$\Scal' \cup \Ccal$ is a partition of $\cup_{S \in \Scal} S$, establishing Items~1(b),(c),(d).

We now turn to establish Item~1(a) of the lemma.  Let 
\[
\Rcal = \{S_i \in \Scal : \exists j \in \{1, \ldots, \ceil{\log |\Scal|}\}, b \in \{0,1\} \mbox{ such that } B_{j,b}(:,i) \ne \vec{0} \} \enspace .
\]
In words, $\Rcal$ is the set of all $S_i$ for which the $i^{th}$ column of some $B_{j,b}$ matrix is nonzero.
By \cref{lem:reduce_high}, $|\Rcal| \le 512 \ceil{\log(|\Scal|)} \ln(n) |\Scal|/d$. 
 We show that for every $U \in \Scal'$ there is 
an $R \in \Rcal$ with $R \subseteq U$.  This will establish Item~1(a) since $ \Scal'$ is a valid set
of supervertices.

By the ``moreover'' statement of \cref{lem:contract}, for every $U \in \Scal'$ there is an $S_i \subseteq U$ with $\text{low}(i) = 0$.  
Thus for some $j,b$ it holds that $S_i \in \Lcal_{j,b}$ and the degree of $S_i$ into $\Rcal_{j,b}$ is at least $d/4$.  
By \cref{lem:reduce_high} the corresponding call to Reduce High on Line~\ref{line:reduce_high} will find a neighbor $R \in \Rcal$ of $S_i$.
Therefore, in the call to Contract a set containing $S_i$ will be merged with a set containing $R$ and therefore $R \subseteq U$.

Finally, the total number of iterations from the two for loops is $O(\log n)$.  Let us now look at the complexity of each iteration 
of the inner for loop.  Approximate Degree Sequence with error probability at most $1/n$ takes $O(\log(n)^3)$ many queries by 
\cref{lem:ADS}.  Each call to Reduce High takes $O(\log(n)^4)$ many queries by \cref{lem:reduce_high}.  Each call to 
Learn Low with error probability $1/n$ takes $O(d \log (n)^2)$ many queries by \cref{lem:learn_low}.  This gives the complexity 
$O(\log(n)^5 + d \log(n)^3)$ as claimed in item~(2).  
\end{proof}

\begin{theorem}
\label{main:con}
Let $G$ be a graph with vertex set $V$ where $|V| =n$.  There is a quantum algorithm that outputs the connected 
components of $G$ with error probability at most $O(\log(n)^2/n)$ after making $O(\log(n)^6)$ many cut queries.  In particular, the algorithm 
determines if $G$ is connected or not with the same number of cut queries.  
\end{theorem}

\begin{proof}
The algorithm is given by \cref{alg:q_con}.  Shrink is called with a degree parameter $d$ of $1024 \ceil{\log(n)}^2$,
thus by item~1(a) of \cref{lem:shrink} the size of the set $\Scal$ will decrease by a factor of at least $2$
in each iteration.  Therefore the number of iterations of the repeat loop will be $O(\log n)$.  This, together with 
item~(3) of \cref{lem:shrink} gives a bound on the total number of queries of $O(\log(n)^6)$.  

Now we argue correctness.  By \cref{lem:shrink}, at each stage of the algorithm we maintain the invariant that $\Scal$ contains connected 
supervertices, $\text{ConComp}$ contains connected components, and $\Scal \cup \text{ConComp}$ is a partition of $V$.  The repeat-until loop will 
terminate as $|\Scal|$ is halved with every iteration.  When it does terminate $\Scal = \emptyset$, thus at this stage $\text{ConComp}$ is a partition of $V$ 
by sets that are connected components. 
\end{proof}

\section{Spanning forest}\label{sec:span}
In this section we show that \cref{alg:q_con} to find connected components can be extended to give an algorithm that 
finds a spanning forest and still only makes polylogarithmically many cut queries.  The key idea for this is to find witnesses 
for superedges found in \cref{alg:q_con}.
\begin{definition}[Witness]
Let $G=(V,E)$ be a graph and $S_i, S_j$ two supervertices of $G$ connected by a superedge.  We say that $\{u,v\}$ is 
a witness for this superedge if $u \in S_i, v \in S_j$ and $\{u,v\} \in E$.
\end{definition}

Let $T_1,T_2$ be spanning trees for connected supervertices $S_1,S_2$, and let $\{u,v\}$ be a witness for a superedge 
between $S_1$ and $S_2$.  Then $T_1 \cup T_2 \cup \{\{u,v\}\}$ is a spanning tree for $S_1 \cup S_2$.  
The spanning forest algorithm proceeds in the same framework
as \cref{alg:q_con}, but now we maintain a spanning tree for each supervertex in $\Scal = \{S_1,\ldots, S_k\}$ and only merge supervertices when we have a witness for a superedge between them.  In this way, we are able to maintain spanning trees for every 
supervertex in $\Scal$ as the algorithm proceeds. The main new difficulty is to find witnesses for the superedges discovered in \cref{alg:q_con}, as there we only discovered the 
existence of an edge between supervertices. However, we show that one can still manage to find enough witnesses to guarantee the size of $\Scal$ shrinks by a factor of $1/2$ in each round.

First we show how to find a witness for each superedge in a bipartite graph if both sides of supervertices have low superdegree.
\begin{algorithm}
\caption{Witness Low-Low$(\Scal,\Tcal,h,\delta)$}
\label{alg:low_low}
 \hspace*{\algorithmicindent} \textbf{Input:} Disjoint valid sets of supervertices $\Scal=\{S_1, \ldots, S_k\}, \Tcal=\{T_1, \ldots, T_\ell\}$, 
 a degree parameter $h$ such that $\deg_{\Tcal}(S_i) \le h$ for all $S_i \in \Scal$ and $\deg_{\Scal}(T_j) \le h$ for all $T_j \in \Tcal$, and an error parameter $\delta$. \\
 \hspace*{\algorithmicindent} \textbf{Output:} A $|\cup_{i=1}^k S_i|$-by-$|\cup_{i=1}^\ell T_i|$ Boolean matrix $C$ such that $C(u,v)=1$ implies $\{u,v\} \in E$ and 
 for every superedge $(S_i,T_j)$ there is a witness $\{u,v\}$ with $C(u,v) = 1$.  
\begin{algorithmic}[1]
\State $U \gets \cup_{i=1}^k S_i, U = \{u_1, \ldots , u_{|U|} \}$ 
\State $Y \gets \cup_{i=1}^\ell T_i$
\State $C \gets \zeros(|U|, |Y|)$
\State $B \gets \text{Learn Low}(\tilde U,\Tcal,h,\delta/2)$ \Comment{Recall $\tilde U = \{\{u\} : u \in U\}$}
\State $X \gets \{ u_a \in U: ( \exists j ~ (B(a,j) > 0) \mbox{ AND } (u_a \in S_i,  u_b \in S_i, b < a ) \Rightarrow  B(b, j) = 0)\}$
\State $D \gets \text{Learn Low}(\tilde Y, \tilde X, h^2, \delta/2)$ 
\State $C(\ind(X), :) = D^T$
\State \Return $C$
\end{algorithmic}
\end{algorithm}

\begin{lemma}[Witness Low-Low]
\label{lem:low_low}
Let $\Scal = \{S_1, \ldots, S_k\}, \Tcal = \{T_1, \ldots, T_\ell\}$ be disjoint valid sets of supervertices in an $n$-vertex graph $G = (V,E)$.  Suppose that $\deg_{\Tcal}(S_i) \le h$ 
for all $S_i \in \Scal$ and $\deg_{\Scal}(T_j) \le h$ for all $T_j \in \Tcal$.  Algorithm Witness Low-Low$(\Scal,\Tcal,h,\delta)$ (\cref{alg:low_low}) makes $O( (h^2 \log(n/h) + \log(1/\delta)) \log (n) )$ 
many cut queries and finds a witness for every superedge between $\Scal$ and $\Tcal$, except with probability $\delta$.
\end{lemma}

\begin{proof}
Let $U = \cup_{i=1}^k S_i$, and $Y = \cup_{i=1}^\ell T_i$.  Let $u_1 < \cdots < u_{|U|}$ be an ordering of the elements of $U$.  
Let $B$ be a $|U|$-by-$\ell$ matrix where $B(a,j) = |E(u_a, T_j)|$ for $a \in \{1, \ldots, |U|\}, j \in \{1, \ldots, \ell\}$.
Note that every row of $B$ has at most $h$ nonzero entries.
By \cref{lem:learn_low}, \text{Learn Low}$(\tilde U, \Tcal, h, \delta/2)$ will return $B$ except with probability at most $\delta/2$, and makes
$O( (h \log(n/h) + \log(1/\delta)) \log (n) )$ many cut queries.  Now 
let $X=\{ u_a \in U: \exists j ~ (B(a,j) > 0 \mbox{ AND } (u_a \in S_i,  u_b \in S_i, b < a ) \Rightarrow  B(b, j) = 0)\}$.  In other words, for every $i,j$ for which there exists
$u_a \in S_i$ with $B(a, j) >0$ we choose the least such $u_a$ to put in the set $X$.
Note that $|X \cap S_i| \le h$ for every $S_i \in \Scal$.  

Let $D$ be the $|Y|$-by-$|X|$ biadjacency matrix of the graph between 
$Y$ and $X$.  Every row will have at most $h^2$ ones, since $\deg_{\Scal}(T_j) \le h$ and $|X \cap S_i| \le h$.  
\text{Learn Low}$(\tilde Y, \tilde X, h^2, \delta/2)$ learns $D$ with 
$O( (h^2 \log(n/h) + \log(1/\delta)) \log (n) )$
many cut queries with error probability $\delta/2$.  
By learning $D$, for every superedge $(S_i,T_j)$ we find a $u \in S_i, v \in T_j$ with $\{u,v\} \in E$.  The total error probability is at most $\delta$.
\end{proof}

Next we show that we can find witnesses for supervertices on the left hand side of a bipartite graph between supervertices where all supervertices 
on the left hand side have low superdegree.
\begin{algorithm}
\caption{Witness Low-High$(\Scal,\Tcal,h)$}
\label{alg:low_high}
 \hspace*{\algorithmicindent} \textbf{Input:} Disjoint valid sets of supervertices $\Scal = \{S_1, \ldots, S_k\}, \Tcal = \{T_1, \ldots, T_\ell\}$, 
 a degree parameter $h$ such that $\deg_{\Tcal}(S_i) \le h$ for all $S_i \in \Scal$. \\
 \hspace*{\algorithmicindent} \textbf{Output:} A $|\cup_i S_i|$-by-$|\cup_j T_j|$ matrix $C$ such that $C(a,b) = 1$ implies $(u_a,v_b) \in E$ and 
 for every $S_i$ with $\deg_{\Tcal}(S_i) > 0$ there is a $u_a \in S_i, v_b \in \cup_j T_j$ such that $C(a,b) = 1$.
\begin{algorithmic}[1]
\State $U \gets \cup_{i=1}^k S_i$
\State $Y \gets \cup_{i=1}^\ell T_i$
\State $C \gets \zeros(|U|,|Y|)$
\State $B \gets \text{Learn Low}(\tilde U,\Tcal,h,1/n)$ \Comment{$\tilde U = \{ \{u\} : u \in U\}$}
\State $X \gets \{ u_a \in U: (\exists j ~ B(a,j) > 0)  \mbox{ AND } (u_a \in S_i, u_b \in S_i,  b < a ) 
\Rightarrow \forall c ~ B(b,c) = 0)\}$
\For{$u \in X$} 
  \State $j^*(u_a) = \argmax_{j \in \{1, \ldots, \ell\}} B(a,j)$
\EndFor
\For{$q = 0, \ldots, \ceil{\log |Y|}$} \label{line:for_q}
  \State $X_q \gets  \{ u \in X: 2^{q-1} < B(u,{j^*(u)}) \le 2^q\}$
  \State $R_q \gets$ Randomly sample $\ceil{16|Y| \ln(|X_q|n)/2^q}$ many elements from $Y$ with replacement
  \State $C(\ind(X_q), \ind(R_q)) \gets \text{Learn Low}(X_q, R_q, 192h \ln(n), 1/n)$
\EndFor
\State \Return $C$
\end{algorithmic}
\end{algorithm}

\begin{lemma}[Witness Low-High]
\label{lem:witness_low_high}
Let $\Scal = \{S_1, \ldots, S_k\}$ and $\Tcal = \{T_1, \ldots, T_\ell\}$ be disjoint valid sets of supervertices in an $n$-vertex graph $G = (V,E)$,
and suppose that $\deg_{\Tcal}(S_i) \le h$ for all $S_i \in \Scal$.  Algorithm Witness Low-High$(\Scal,\Tcal,h)$ (\cref{alg:low_high})
makes $O(h \log(n)^4)$ many cut queries and except with probability $O(\log(n)/n)$ outputs a $|\cup_i S_i|$-by-$|\cup_j T_j|$ matrix $C$ such 
that 
\begin{enumerate}
 \item $C(a,b) = 1$ implies $\{u_a, v_b\} \in E$ 
 \item For every $S_i$ with $\deg_{\Tcal}(S_i) > 0$ there is a $u_a \in S_i, v_b \in \cup_j T_j$ such that $C(a,b) = 1$. 
\end{enumerate}
\end{lemma}
 
 \begin{proof}
Let $U = \cup_{i=1}^k S_i$ and let $u_1 < \cdots < u_{|U|}$ be an ordering of the elements of $U$.  The algorithm first performs $B \gets \text{Learn Low}(\tilde U,\Tcal,h,1/n)$.  
By the correctness of \text{Learn Low} from \cref{lem:learn_low}, 
except with probability $1/n$, it will hold that $B(a, j) = |E(u_a, T_j)|$ for all $a \in \{1, \ldots, |U|\}, j \in \{1, \ldots, \ell\}$.  This step takes $O(h \log(n/h) \log (n))$ cut queries.  We henceforth 
assume that this step was performed correctly.
Define $X = \{ u_a \in U: (\exists j ~ B(a,j) > 0)  \mbox{ AND } ((u_a \in S_i, u_b \in S_i,  b < a ) \Rightarrow \forall c ~ B(b,c) = 0)\}$.
In other words, for every $S_i$ we take the first $u_a \in S_i$ for which there is a $j \in \{1, \ldots, \ell\}$ such that $B(a,\ell) > 0$, if such a $u_a$ exists.
Let $j^*(u_a) = \argmax_{j \in \{1, \ldots, \ell\}} B(a,j)$ for each $u_a \in X$.  
For each $u \in X$, we are going to learn a $v \in T_{j^*(u)}$ such that $\{u,v\} \in E$.

Fix a value of $q$ in the for loop starting on line~\ref{line:for_q}.  For every $u_a \in X_q$ and $j \in \{1, \ldots, \ell\}$ we have $B(a,j) \le B(a,j^*(u_a)) \le 2^q$.  
For $\delta = 1/n^2$, $R_q$ is formed by randomly sampling with replacement $\ceil{8|Y| \ln(|X_q|/\delta)/2^q}$ elements of $Y$. Thus by \cref{lem:sample} 
and a union bound over $j \in \{1, \ldots, \ell\}$, except with probability $1/n$, we have
\begin{equation}
\label{eq:upper_j}
|E(u, R_q \cap T_j)|  \le 192 \ln(n) \mbox{ for all } j \in \{1,\ldots, \ell\} \mbox{ and } u \in X_q   \enspace.
\end{equation}
Similarly, except with probability $1/n$, we also have 
\begin{equation*}
|E(u, R_q \cap T_j)|  > 0 \mbox{ for all } j \in \{1,\ldots, \ell\} \mbox{ and } u \in X_q \mbox{ with } 2^{q-1} < B(u,j) \enspace.
\end{equation*}
In particular, except with probability $1/n$
\begin{equation}
\label{eq:lower_j}
|E(u, R_q \cap T_j^*(u))|  > 0 \mbox{ for all } u \in X_q \enspace .
\end{equation}
We now add $O(1/n)$ to the error probability and assume for the rest of the proof that \cref{eq:upper_j} and \cref{eq:lower_j} hold.
Next the algorithm performs
\[
C(\ind(X_q), \ind(R_q)) \gets \text{Learn Low}(X_q, R_q, 192h \ln(n), 1/n) \enspace.
\]
By \cref{eq:upper_j}, the upper bound on the degree passed to Learn Low is valid, thus by \cref{lem:learn_low} Learn Low 
returns the biadjacency matrix between $X_q$ and $R_q$, except with probability $1/n$.  In particular, except with probability 
$1/n$, if $C(a,b) = 1$ for $a \in \ind(X_q), b \in \ind(R_q)$ then $\{u_a, v_b\} \in E$ and 
by \cref{eq:lower_j} for every $u_a \in X_q$ there is $v_b \in R_q$ such that $C(a,b) = 1$.
This step requires $O(h \log (n)^2 \log(n/h))$ many cut queries.

Finally, by a union bound over $q=0, \ldots, \ceil{\log |Y|}$, the probability of an error in any iteration of the for loop is at most $O(\log(n)/n)$.
As every $u \in X$ will be in $X_q$ for some value of $q$, for every $S_i$ with $\deg_T(S_i) > 0$ we will find a $u \in S_i$ and 
 $v \in Y$ with $\{u,v\} \in E$.  The total error probability is $O(\log(n) /n)$, and total number of queries is $O(h \log (n)^4)$.
 \end{proof}

Next we show how to find witnesses for supervertices of high superdegree.
\begin{algorithm}
\caption{Witness Reduce High$(\Scal,\Tcal,d,\vec{g})$}
\label{alg:witness_reduce_high}
 \hspace*{\algorithmicindent} \textbf{Input:} Disjoint valid sets of supervertices 
 $\Scal = \{S_1, \ldots, S_k\}, \Tcal = \{T_1, \ldots, T_\ell\}$ of an $n$-vertex graph $G$, a degree parameter $d$ such that 
 $\deg_\Tcal(S_i) \ge d$ for all $S_i \in \Scal$, and a vector 
 $\vec{g} \in \mathbb{Z}^k$ such that $\vec{g}(i)i/4 \le \deg_{\Tcal}(S_i) \le 2 \vec{g}(i)$ for all $i \in \{1, \ldots, k\}$. \\
  \hspace*{\algorithmicindent} \textbf{Output:} A $|\cup_{i=1}^k S_i|$-by-$|\cup_{i=1}^\ell T_i|$ Boolean matrix $C$ such that $C(u,v)=1$ implies 
  $\{u,v\} \in E$ and for every $S_i$ there is a $u \in S_i, v \in \cup_j T_j$ such that $C(u,v) = 1$.
\begin{algorithmic}[1]
\State $C \gets \zeros(|\cup_{i=1}^k S_i|, |\cup_{i=1}^\ell T_i|)$
\For{$q = \floor{\log d}-1$ to $\ceil{\log \ell}+2$}
  \State $\Hcal_q = \{ S_i \in \Scal : 2^{q-1} < \vec{g}(i) \le 2^{q}\}$ 
  \State $\Rcal_q \gets$ Randomly choose $\frac{16 \ell \ln(|\Hcal_q| n)}{2^q}$ supervertices in $\Tcal$, with replacement
  \State $C(\ind(\Hcal_q), \ind(\Rcal_q))  \gets$ Witness Low-High$(\Hcal_q,\Rcal_q, 192 \ln(n))$ \label{li:rh_lh}
\EndFor
\State \Return $B$
\end{algorithmic}
\end{algorithm}

\begin{lemma}[Witness Reduce High]
\label{lem:witness_reduce_high}
Let $\Scal = \{S_1, \ldots, S_k\}, \Tcal = \{T_1, \ldots, T_\ell\}$ be disjoint valid sets of supervertices of an $n$-vertex graph $G$.  Let $d \in \N$ be 
a degree parameter and $\vec g \in \R^k$ be a good estimate of $(\deg_{\Tcal}(S_1), \ldots, \deg_{\Tcal}(S_k))$.  
Further suppose that $\deg_{\Tcal}(S_i) \ge d$ for all $S_i \in \Scal$.
On input $\Scal,\Tcal,d, \vec{g}$, Witness Reduce High (\cref{alg:witness_reduce_high}) makes $O(\log(n)^5)$ many cut queries and with probability at least 
$1-O(\log(n)^2/n)$ outputs a $|\cup_i S_i|$-by-$|\cup_j T_j|$ Boolean matrix $C$ such that 
\begin{enumerate}
\item For every $S_i \in \Scal$, there exists $u \in S_i$ and $v \in \cup_j T_j$ such that $C(u,v)=1$.
\item If $C(u,v) = 1$ then $\{u,v\} \in E$.
\item $|\{T_j : \exists u \in \cup_i S_i, \exists v \in T_j : C(u,v) =1\}| \le \frac{256 \ell \log(n)}{d}$.
\end{enumerate}  
\end{lemma}

\begin{proof}
We first show item~(2).  By \cref{lem:witness_low_high}, except with probability $O(\log(n)/n)$, the call to Witness Low-High on  
Line~\ref{li:rh_lh} will only return valid edges in $G$.  As there are $O(\log n)$ iterations of the for loop, item~(2) will therefore hold 
except with probability $O(\log(n)^2/n)$.

We now show items~(1) and (3), assuming that all calls to Witness Low-High return correctly.  The only difference between Reduce High and Witness Reduce High is that in 
Line~\ref{li:rh_lh} of Witness Reduce-High we call Witness Low-High instead of Learn Low.  As Witness Low-High returns correctly, by \cref{lem:witness_low_high} this means  
we find a witness for every superedge found by Learn Low, and therefore a witness for every superedge found by Reduce High.  By Item~(1) of \cref{lem:reduce_high} 
for every $S_i \in \Scal$ Reduce High finds a $T_j \in \Tcal$ such that $(S_i, T_j)$ is a superedge.  Finding a witness for each of these superedges gives Item~(1) here.  

Furthermore, we do not find witnesses for any superedges not found in Reduce High.  Thus Item~(3) of \cref{lem:reduce_high} implies Item~(3) here.

Finally, we bound the number of queries made.  There are $O(\log n)$ iterations of the for loop, and in each iteration queries are only 
made in the call to Witness Low-High.  Each of these calls take $O(\log(n)^4)$ cut queries, thus the total number of cut queries is 
$O(\log(n)^5)$.
\end{proof}

Next we show how to contract supervertices by using the edges found so far.
\begin{lemma}
\label{lem:witness_contract}
Let $\Scal = \{S_1, \ldots, S_k\}$ be a valid set of supervertices with $N = |\cup_{i=1}^k S_i|$, $\Pcal = \{P_1, \ldots, P_k\}$ a set of spanning trees, 
$\Acal$ a list of $N$-by-$N$ weighted adjacency matrices, and $low \in \{0,1\}^k$ a Boolean vector with the following properties:
\begin{enumerate}
  \item Every supervertex in $\Scal$ is connected, and a spanning tree for $S_i \in \Scal$ is given by $P_i \in \Pcal$.
  \item For every $A \in \mathcal{A}$ if $A(u,v) = 1$ then $\{u,v\} \in E$.
  \item For every $i \in \{1,\ldots, k\}$ with $low(i) = 1$, it holds that for every $j \in \{1, \ldots, k\}$ such that there is a superedge between $S_i$ and $S_j$, 
  there is a $u \in S_i, v \in S_j$ and $A \in \mathcal{A}$ with $A(u,v) = 1$.
\end{enumerate}
The algorithm Witness Contract$(\Scal, \Pcal, \mathcal{A}, low)$ given in \cref{alg:witness_contract} makes no queries and outputs sets of supervertices $\Scal' = \{S_1', \ldots, S_\ell'\}, \Ccal = \{C_1, \ldots, C_t\}$ and sets of spanning 
trees $\Pcal_{\Scal'} = \{P_1', \ldots, P_\ell'\}, \Pcal_{\Ccal} = \{Q_1, \ldots, Q_t\}$ such that 
\begin{itemize}
  \item Each $S_i' \in \Scal'$ is connected and a spanning tree for it is given by $P_i' \in \Pcal_{\Scal'}$.
  \item Each $C_i \in \Ccal$ is a connected component and a spanning tree for it is given by $Q_i \in \Pcal_{\Ccal}$.
  \item $\Scal' \cup \Ccal$ is a partition of $\cup_{S \in \Scal} S$.  
  \item For every $U \in \Scal'$ there is an $S_i \subseteq U$ with $low(i) = 0$.
\end{itemize}
\end{lemma}

\begin{algorithm}
\caption{Witness Contract$(\Scal, \Pcal, \mathcal{A}, low)$}
\label{alg:witness_contract}
 \hspace*{\algorithmicindent} \textbf{Input:} Valid set of connected supervertices $\Scal = \{S_1, \ldots, S_k\}$, a set of spanning trees $\Pcal = \{P_1, \ldots, P_k\}$, 
 a list of $N$-by-$N$ weighted adjacency matrices 
 $\mathcal{A}=(A_1, \ldots, A_m)$ with rows and columns labeled by elements of $[N]$, a vector $low \in \{0,1\}^k$ indicating if each set $S_i$ is low. \\
 \hspace*{\algorithmicindent} \textbf{Output:} Sets of supervertices $\Scal', \Ccal$ and sets of spanning trees for them $\Pcal_{\Scal'}, \Pcal_{\Ccal}$, where each supervertex in $\Scal', \Ccal$ is connected, and moreover the 
 supervertices in $\Ccal$ are connected components.
\begin{algorithmic}[1]
\State $L \gets [ (S_i, P_i, low(i) ) : i \in \{1,\ldots, k\}]$
\For{$A \in \Acal$}
  \For{$\{u,v\} \in [N]^{(2)}$}
    \If{$A(u,v) = 1$}
      \State Pop $(U, T_1, flag1) \in L$ such that $u \in S_i \subseteq U$
      \State Pop $(W, T_2, flag2) \in L$ such that $v \in S_j\subseteq W$
      \State $U \gets U \cup W$
      \State $T = T_1 \cup T_2 \cup \{\{u,v\}\}$
      \State $lowFlag = flag1 \wedge flag2$
      \State Append $(U, T, lowFlag)$ to $L$
    \EndIf
  \EndFor
\EndFor
\State $\Scal' \gets \emptyset, \Pcal_{\Scal'} \gets \emptyset$
\State $\Ccal \gets \emptyset, \Pcal_{\Ccal} \gets \emptyset$
\For{$(U, T, lowFlag) \in L$}
  \If{lowFlag = 0}
    \State $\Scal' \gets \Scal' \cup \{U\}$ 
    \State $\Pcal_{\Scal'} \gets \Pcal_{\Scal'} \cup \{T\}$ 
  \Else
    \State $\Ccal \gets \Ccal \cup \{U\}$ \label{line:ccal}
    \State $\Pcal_{\Ccal} \gets \Pcal_{\Ccal} \cup \{T\}$ 
  \EndIf
\EndFor
\State \Return $\Scal', \Pcal_{\Scal'}, \Ccal, \Pcal_{\Ccal}$  
\end{algorithmic}
\end{algorithm}
\begin{algorithm}
\caption{Witness Shrink$(\Scal, \Pcal,d)$}
\label{alg:witness_shrink}
 \hspace*{\algorithmicindent} \textbf{Input:} A valid set of connected supervertices $\Scal =\{S_1, \ldots, S_k\}$, a set of spanning trees 
 $\Pcal = \{P_1, \ldots, P_k\}$ where $P_i$ is a spanning tree for $S_i$, and a degree parameter $d$. \\
 \hspace*{\algorithmicindent} \textbf{Output:} Sets of supervertices $\Scal, \Ccal$ and corresponding sets of spanning trees $\Pcal_{\Scal}, \Pcal_{\Ccal}$.
\begin{algorithmic}[1]
\State $low  \gets \ones(k,1)$
\State $\Acal \gets [\;]$
\State $N \gets | \cup_{i=1}^k S_i|$
\For{$j = 1$ to $\ceil{\log(k)}$}
  \For{$ b \in \{0,1\}$}
  \State $\Lcal_{j,b} = \{S_t \in \Scal : t_j = b\}$
  \State $\Rcal_{j,b} = \{S_t \in \Scal : t_j = 1-b\}$
  \State $\vec{g} \gets$ Approximate Degree Sequence$(\Lcal_{j,b}, \Rcal_{j,b}, 1/n^2)$
  \State $\Hcal \gets \{S_t \in \Lcal_{j,b} : \vec{g}(t) \ge d\}$
  \State $low(\ind(\Hcal)) = 0$
  \State $B_{j,b} \gets \zeros(N,N)$
  \State $B_{j,b}(\cup_{S_t \in \Hcal} S_t,\cup_{S_t \in \Rcal_{j,b}} S_t) \gets$ \text{Witness Reduce High} $(\Hcal, \Rcal_{j,b},d/4, \vec{g}(\ind(\Hcal)))$ \label{li:witness_rh}
  \State $\Lcal \gets \{S_t \in \Lcal_{j,b} : \vec{g}(t) < d\}$
  \State $\vec{f} \gets$ Approximate Degree Sequence$(\Rcal_{j,b}, \Lcal, 1/n^2)$
  \State $R_{j,b}^+ \gets \{ S_t \in \Rcal_{j,b} : \vec{f}(t) \ge 16 d\}$
  \State $R_{j,b}^- \gets \{ S_t \in \Rcal_{j,b} : \vec{f}(t) < 16 d\}$
  \State $C_{j,b}^+(\cup_{S_t \in \Lcal} S_t, \cup_{S_t \in \Rcal_{j,b}^+} S_t) \gets$ Witness Low-High$(\Lcal, \Rcal_{j,b}^+, 2d)$
  \State $HasHighNeighbor \gets \{i : \exists u \in S_i \in \Lcal, \exists v \colon C_{j,b}^+(u,v) = 1\}$
  \State $low(HasHighNeighbor) = 0$
  \State $C_{j,b}^-(\cup_{S_t \in \Lcal} S_t, \cup_{S_t \in \Rcal_{j,b}^-} S_t) \gets$ Witness Low-Low$(\Lcal, \Rcal_{j,b}^-,32d,1/n)$
  \State Append $B_{j,b}, C_{j,b}^+,$ and $C_{j,b}^-$ to $\Acal$.
  \EndFor
\EndFor
  \State $\Scal, \Pcal, \Ccal, \Pcal_{\Ccal} \gets$ Witness Contract$(\Scal,\Pcal, \Acal)$ \label{line:witness}
  \State \Return $\Scal, \Pcal, \Ccal, \Pcal_{\Ccal}$
\end{algorithmic}
\end{algorithm}

\begin{proof}
We first show that each $S_i' \in \Scal'$ is connected and a spanning tree for it is given by $P_i' \in \Pcal_{\Scal}$.
By assumption each $S_i \in \Scal$ is connected with a spanning tree given by $P_i \in \Pcal$.  We maintain this 
invariant because we only merge two supervertices $U,W$ when there is a $u \in U, w \in W$ and a $A \in \Acal$ 
with $A(u,w) = 1$.  By assumption if $A(u,w) = 1$ then $\{u,w\} \in E$ and thus $U \cup W$ is connected when 
$U,W$ are.  Furthermore, if $T_1$ is a spanning tree for $U$ and $T_2$ is a spanning tree for $W$ then 
$T_1 \cup T_2 \cup \{ \{u,w\}\}$ is a spanning tree for $U \cup W$.  

The same argument shows that each $C_i \in \Ccal$ is connected with a spanning tree given by $P_i \in \Pcal_{\Ccal}$.

Next we show that every $W \in \Ccal$ is a connected component.  Suppose for a contradiction that this is not the case and therefore there 
is a $w \in S_i \subseteq W$ and $u \in S_j \subseteq (\cup_{t=1}^k S_t) \setminus W$ such that $\{u,w\}$ is an edge of $G$.  It must be the 
case that $S_i$ is low, as otherwise the lowFlag for $W$ would have been set to $0$ and $W$ would have been placed in $\Scal'$.  Thus 
$S_i$ must be low and therefore by hypothesis for some $w' \in S_i, u' \in S_j$ and $A \in \Acal$ it is the case that $A(w',u') = 1$.  This means that at some point in 
Witness Contract a set containing $S_i$ would have been merged with a set containing $S_j$, a contradiction to the fact that 
$S_j \subseteq (\cup_{t=1}^k S_t) \setminus W$.

The fact that $\Scal \cup \Ccal$ is a partition of $\cup_{S \in \Scal} S$ follows because at all times $\cup_{(U,T,flag) \in L} U$ is equal 
to $\cup_{S \in \Scal} S$.  This is true when $L$ is first defined, and is preserved when sets are popped from $L$, merged, and put 
back into $L$.

Finally, the ``moreover'' statement holds as if $low(i) = 1$ for all $S_i \in U$ then the lowFlag variable for $U$ will be set to $1$ and 
therefore $U$ will be placed into $\Ccal$ on Line~\ref{line:ccal}.
\end{proof}

The last lemma guarantees the shrinkage of the number of supervertices.
\begin{lemma}
\label{lem:witness_shrink}
Let $\Scal =\{S_1, \ldots, S_k\}$ be a valid set of connected supervertices, $\Pcal =\{P_1, \ldots, P_k\}$ be a set of spanning trees where $P_i$ is a 
spanning tree for $S_i$, and $d \in \mathbb{N}$ be a degree parameter.  Witness Shrink$(\Scal, \Pcal,d)$ given by \cref{alg:witness_shrink} has the following properties:
\begin{enumerate}
\item Except with probability $O(\log(n)^3/n)$, it outputs sets of supervertices $\Scal', \Ccal$ and sets of spanning trees $\Pcal_{\Scal'}, \Pcal_{\Ccal}$ such that
\begin{enumerate}
  \item $|\Scal'| \le 512 \ceil{\log(|\Scal|)} \ln(n) |\Scal |/d$.  
  \item For every supervertex in $\Scal'$ there is a spanning tree for it in $\Pcal_{\Scal'}$.
  \item Every supervertex in $\Ccal$ is a connected component and has a spanning tree for it in $\Pcal_{\Ccal}$.
  \item $\Scal' \cup \Ccal$ is a partition of $\cup_{S \in \Scal} S$.  
\end{enumerate}
  \item The total number of cut queries made is $O(d\log(n)^5 + d^2 \log(n)^3)$.  
\end{enumerate}
\end{lemma}

\begin{proof}
We first prove item~(2).  Queries are only made in the calls to Approximate Degree Sequence, Witness Reduce High,
Witness Low-High, and Witness Low-Low.  Each of these routines is called $2 \ceil{\log|\Scal|} = O(\log n)$ many times.  The query cost is 
dominated by Witness Low-High, which takes $O(d \log(n)^4)$ many queries, and Witness Low-Low which requires 
$O(d^2 \log(n)^2)$ queries, resulting in a total of $O(d\log(n)^5 + d^2 \log(n)^3)$ cut queries.  

The error probability of the call to Approximate Degree Sequence is $1/n^2$, Witness Reduce High is $O(\log(n)^2/n)$, 
Witness Low-High is $O(\log(n)/n)$, and Witness Low-Low is $1/n$.  Thus the probability that an error occurs in any 
of these routines over the course of Witness Shrink is $O(\log(n)^3/n)$.  We now argue the points in item~(1) hold assuming all 
of these routines always return correctly.  

We first establish that the hypotheses of \cref{lem:witness_contract} hold when Witness Contract is called on Line~\ref{line:witness}.
By assumption the supervertices in $\Scal$ are connected and spanning trees for them are given in $\Pcal$.  The matrices 
in $\Acal$ are produced in calls to Witness Reduce High, Witness Low-High, and Witness Low-Low.  As we are in the case 
that all of these algorithms return correctly it follows by \cref{lem:low_low},\cref{lem:witness_low_high}, and \cref{lem:witness_reduce_high} 
that for every $A \in \Acal$ if $A(u,v) = 1$ then $\{u,v\} \in E$.  
Finally, we need to establish that if $low(i) = 1$ then $\Acal$ contains a witness for every superedge of $S_i$.  Suppose that $low(i) = 1$ and 
$S_i$ has a superedge with $S_t$.  For some value of $j,b$ in the for loop we will have $S_i \in \Lcal_{j,b}, S_t \in \Rcal_{j,b}$.  As $low(i)=1$ 
and Witness Low-High returns correctly, $S_i$ has no neighbors in $\Rcal_{j,b}^+$.  Thus it must be the case that $S_t \in \Rcal_{j,b}^-$.  As 
$\vec f, \vec g$ are good estimates because Approximate Degree Sequence returns correctly, the degree bound in the call to Witness Low-Low is 
valid and therefore by \cref{lem:low_low} a witness for the $(S_i,S_t)$ superedge will be found in the call to Witness Low-Low.

We have now established the hypotheses to \cref{lem:witness_contract} and thus can invoke the conclusion of \cref{lem:witness_contract} 
which implies Items~1(b),(c),(d) of the current lemma.

It remains to show Item~1(a).  Let 
\[
\Rcal = \{S_t \in \Scal : \exists j \in \{1, \ldots, \ceil{\log k}\}, b \in \{0,1\}, v \in S_t
\mbox{ such that } B_{j,b}(:,v) \ne \vec{0} \} \enspace .
\]
By \cref{lem:witness_reduce_high}, and the fact that the number of iterations of the for loop is $2 \ceil{\log |\Scal| }$, 
we have $|\Rcal| \le 512 \ceil{\log |\Scal| } \log(n) |\Scal |/d$.  In Witness Contract, a supervertex $W$ will be put into 
$\Scal'$ iff for some $S_i \in W$ we have $low(i) = 0$.    
In the next paragraph 
we show that 
in that case
Witness Shrink finds witnesses to certify that an element of $\Rcal$ is in the connected component 
of $S_i$.  This means that Witness Contract will merge $S_i$ into a set $W$ containing an element of $\Rcal$ and 
therefore the number of supervertices in $\Scal'$ can be upper bounded by $|\Rcal|$ and will give Item~1(a).

Take an $S_i$ with $low(i)=0$ and consider the iteration $j,b$ of the for loop where $low(i)$ is set to zero.  There are two ways this can happen.  The first is 
if $\vec g(S_i) \ge d$.  In this case $S_i$ will be placed into $\Hcal$ and a witness for a neighbor in $\Rcal$ will be found in the
call to Witness Reduce High by \cref{lem:witness_reduce_high}.  The second case is that a witness is found for a superedge 
between $S_i$ and an element $S_\ell \in \Rcal_{j,b}^+$.  In the $j,1-b$ iteration of the for loop, $S_\ell \in \Lcal_{j,1-b}$ 
and moreover $\vec g(\ell) \ge d$, as $\vec f(\ell) \ge 16d$ for $S_\ell$ to be placed in $\Rcal_{j,b}^+$ and both are good estimates.  Therefore by the 
previous argument a witness for a neighbor of $S_\ell$ with an element of $\Rcal$ will be found in the call to Witness Reduce High.  Thus we have 
witnesses that $S_i$ is connected to $S_\ell$ and that $S_\ell$ is connected to an element of $\Rcal$.  
\end{proof}

Finally we can give the algorithm of finding a spanning forest with 
polylogarithmic many cut queries.
\begin{algorithm}
\caption{Spanning Forest with cut queries}
\label{alg:span_forest}
 \hspace*{\algorithmicindent} \textbf{Input:} Cut oracle for a graph $G$ on $n$ vertices. \\
 \hspace*{\algorithmicindent} \textbf{Output:} A set $\Pcal = \{P_1, \ldots, P_t\}$ containing a spanning tree for every connected component of $G$.
\begin{algorithmic}[1]
\State $\Scal \gets \tilde V, \Pcal_{\Scal} \gets \tilde V$
\State $\Pcal \gets \emptyset, \text{ConComp} \gets \emptyset$
\Repeat
  \State $(\Scal, \Pcal_{\Scal}, \Ccal, \Pcal_\Ccal) \gets WitnessShrink(\Scal,\Pcal_{\Scal}, 1024 \ceil{\log n}^2)$
  \State $\text{ConComp} \gets \text{ConComp} \cup \Ccal$
  \State $\Pcal \gets \Pcal \cup \Pcal_\Ccal$
\Until{$\Scal = \emptyset$}
\State \Return $\Pcal$
\end{algorithmic}
\end{algorithm}

\begin{theorem}
\label{thm:span_forest}
Given cut query access to an $n$-vertex graph $G$, there is quantum algorithm (\cref{alg:span_forest}) making $O(\log(n)^8)$ queries
that outputs a spanning forest for $G$ with probability $1-O(\log(n)^4/n)$.
\end{theorem}

\begin{proof}
We first argue by induction that the $i^{th}$ time Witness Shrink is called the hypothesis to \cref{lem:witness_shrink} is satisfied
with probability at least $1- (i-1) \log(n)^3/n$.  

The first time Witness Shrink is called, $\Scal = \tilde V$ and $\Pcal_{\Scal} = \tilde V$.  Thus $\Pcal_{\Scal}$ provides 
valid spanning trees for each supervertex in $\Scal$ and the hypothesis to \cref{lem:witness_shrink} is satisfied.  Now suppose the inductive
assumption holds the $i^{th}$ time Witness Shrink is called.  Then after the call to Witness Shrink we know that except with probability 
$\log(n)^3/n$ the output $\Scal, \Pcal$ satisfy that each supervertex in $\Scal$ is connected and has a valid spanning tree for it given in 
$\Pcal$.  Thus in the $(i+1)^{th}$ call to Witness Shrink the hypothesis to \cref{lem:witness_shrink} holds with probability at least $1- i \log(n)^3/n$.

If the hypothesis to \cref{lem:witness_shrink} holds in the call to Witness Shrink, then by the choice of the degree parameter $d = 1024 \ceil{\log n}^2$
the size of $\Scal$ will reduce by a factor of $1/2$ with every iteration of the repeat until loop.  Thus with probability at least $1- \log(n)^4/n$ every 
call to Witness Shrink will return correctly and the number of iterations will be at most $O(\log(n))$.  In this case the total number of queries made 
will be $O(\log(n)^8)$.

If every call to Witness Shrink returns correctly, then the algorithm mantains the invariant that $\Scal \cup \text{ConComp}$ is a partition 
of $V$, every supervertex in $\text{ConComp}$ is a connected component and has a valid spanning tree in $\Pcal_{\Ccal}$.  At the end of the 
algorithm $\Scal = \emptyset$ thus $\text{ConComp}$ is contains all connected components of $G$ and $\Pcal_{\Ccal}$ is a spanning forest.
\end{proof}

\subsection{Applications}
With the ability to compute a spanning forest we can also easily solve some other graph problems in the cut query model.

\begin{theorem}
\label{thm:bipartite}
There is a quantum algorithm to determine if an $n$-vertex graph $G$ is bipartite that makes $O(\log(n)^8)$ many cut queries 
and succeeds with probability at least $1-O(\log(n)^4/n)$.
\end{theorem}

\begin{proof}
We first invoke \cref{thm:span_forest} to find a spanning forest for $G$ with $O(\log(n)^8)$ many cut queries and success 
probability $1-O(\log(n)^4/n)$.  We then color each root of a spanning tree red, and proceed to color all the remaining vertices 
blue and red such that no two vertices connected in a spanning tree have the same color.  The graph is then bipartite 
if and only if there is no edge of $G$ between two vertices of the same color.  

We can check if there is an edge between two red vertices with $O(\log(n))$ cut queries.  Let $S$ be the set of red vertices.  We 
consider $O(\log n)$ bipartite graphs $(L_i, R_i, E_i)$ where $L_i = \{v \in S: v_i= 0\}, R_i = \{v \in S: v_i= 1\}$ and $E_i = \{(u,v) : u \in L_i, v \in R_i, \{u,v\} \in E\}$.
Then with 3 cut queries we can check if $|E(L_i, R_i)| > 0$.  An edge between two red vertices will be present in at least one of these 
bipartite graphs, thus this process will determine if there is an edge between two red vertices.  We then do the same procedure for the 
blue vertices.
\end{proof}

Similarly we can check if a graph is acyclic.  
\begin{theorem}
There is a quantum algorithm to determine if an $n$-vertex graph $G$ is acyclic that makes $O(\log(n)^8)$ many cut queries 
and succeeds with probability at least $1-O(\log(n)^4/n)$.
\end{theorem}

\begin{proof}
We first check that the graph is bipartite, i.e.\ that it has no odd cycles, using \cref{thm:bipartite}.  If it is bipartite, then it remains 
to check that it also has no even cycles.

To check for even cycles we do the same procedure as in the proof of \cref{thm:bipartite}: we find a spanning forest and color the vertices of the spanning 
trees red and blue.  Then the graph will have no even cycles if and only if there are no additional edges between red and blue vertices than those present 
in the spanning trees.  Let the set of red vertices be $R$ and the set of blue vertices be $B$.  With three cut queries we determine $|E(B,R)|$.  We then 
compare this to the number of edges between blue and red vertices that in the spanning forest.  There is no even cycle if and only if these numbers are 
the same.
\end{proof}

\paragraph{Remark} In both applications, after finding a spanning forest, the problem essentially becomes testing graph emptiness: In bipartite testing we need to check the vertices of the same color form an empty graph, and in acyclic graph testing we need to check that there is no edge other than those in the found spanning forest. If a small constant error is tolerated (as opposed to the $\tilde O(1/n)$ one obtained in the above two proofs), then testing graph emptiness can be done in a constant number of queries. Indeed, a query $E(S,V\setminus S)$ for a random subset $S\subseteq V$  returns a positive integer as long as one edge exists, and repeating this $\lceil \log(1/\epsilon)\rceil $ times gives an error probability of at most $\epsilon$.

\section{Concluding remarks}\label{sec:conclusion}
In this paper we investigate the power of additive and cut queries on graphs, and demonstrate that quantum algorithms using these oracles that can solve certain graph problems with surprisingly low query cost. Some open questions are left for future 
investigation, and we list a few of them here.
\begin{enumerate}
    \item The most pressing problem left open by this work is the quantum complexity of minimum cut with a cut oracle.  Can this be solved with a polylogarithmic number of queries?
    \item Classically, the best known  lower bounds on the query complexity of minimizing a submodular function with an evaluation oracle can be shown via the connectivity problem.  We have ruled connectivity out as a candidate for 
    a good quantum lower bound, and in fact we do not know of any nontrivial lower bound on the quantum query complexity of minimizing a submodular function.  As a modest challenge, can one give an example of a submodular function 
    whose minimization problem requires $\Omega(n)$ many evaluation queries by a quantum algorithm?
    \item Another interesting problem is the maximization of a submodular function.  This problem is NP-hard in general and classically exponentially large query lower bounds are known.  For example, \cite{FMV11} show that 
    $\exp(\epsilon^2 n/8)$ many evaluation oracle queries can be needed by a randomized algorithm even to find a $(\frac{1}{2}+\epsilon)$-approximation to the maximum value of a submodular function.  What is the quantum 
    query complexity of submodular function maximization?
\end{enumerate}

\section*{Acknowledgments}
We would like to thank Jon Allcock for helpful discussions during the course of this work and Tongyang Li for useful comments on an earlier version of the paper.
M.S. thanks Yassine Hamoudi  for helpful conversations on submodular function oracles. 
T.L. thanks the Centre for Quantum Technologies, Singapore for supporting a visit where this work began.  T.L. is supported in part by the Australian Research Council Grant No: DP200100950.  
Research at CQT is funded by the National Research Foundation, the Prime Minister's Office, and the Ministry of Education, Singapore under the Research Centres of Excellence programme's research grant R-710-000-012-135. 
In addition, this work has been supported in part by the QuantERA ERA-NET Cofund project QuantAlgo and the ANR project ANR-18-CE47-0010 QUDATA. 

\newcommand{\etalchar}[1]{$^{#1}$}

\end{document}